%% file: main.tex
\documentclass[runningheads]{llncs}
\usepackage[T1]{fontenc}

\usepackage{multirow}   
\usepackage{booktabs}    

\usepackage{hyperref}
\usepackage{graphicx}
\usepackage{amssymb}
\usepackage{amsmath}  
\usepackage{bm}
\usepackage{graphicx}
\usepackage{color}
\usepackage{svg}
\usepackage{subcaption}
\usepackage{threeparttable}
\usepackage{marvosym}
\DeclareMathOperator*{\rank}{rank}

\DeclareMathOperator*{\DET}{DET}
\DeclareMathOperator*{\LP}{LP}

\DeclareMathOperator*{\cols}{cols}

\definecolor{mygreen}{RGB}{123, 211, 123}

\usepackage{algorithm}
\usepackage{algorithmic}
\usepackage{xcolor}

\usepackage{graphicx} 
\usepackage{hyperref} 
\usepackage[firstpage]{draftwatermark} 
\usepackage{orcidlink}
\hypersetup{hidelinks} 
\begin{document}
%



\title{Inner-approximate Reachability Computation via Zonotopic Boundary Analysis
} 

\subtitle{}

\author{Dejin Ren\inst{1,2}\textsuperscript{(\Letter)}\orcidlink{0000-0001-7779-0096} \and
 Zhen Liang\inst{3}\orcidlink{0000-0002-1171-7061} \and
Chenyu Wu\inst{1,2}\orcidlink{0009-0003-1571-4831}
\and
Jianqiang Ding\inst{4}\orcidlink{0000-0003-0705-0345}
\and
Taoran Wu\inst{1,2}\orcidlink{0000-0003-3398-0466}\and
Bai Xue\inst{1,2}\textsuperscript{(\Letter)}\orcidlink{0000-0001-9717-846X}
}
\authorrunning{D. Ren et al.}
%
\institute{
Key Laboratory of System Software (Chinese Academy of Sciences)\\ and State Key Laboratory of Computer Science, Institute\\ of Software, Chinese Academy of Sciences, Beijing, China \\
\email{\{rendj, wucy, wutr, xuebai\}@ios.ac.cn}
\and
University of Chinese Academy of Sciences, Beijing, China
\and
College of Computer Science and Technology,
National \\University of Defense Technology, Changsha, China\\
\email{liangzhen@nudt.edu.cn}
\and
Department of Electrical Engineering and Automation,\\ Aalto University, Espoo, 
Finland
\\
\email{jianqiang.ding@aalto.fi}
}
%


\maketitle              

\SetWatermarkAngle{0}
\SetWatermarkText{\raisebox{10cm}{%
\hspace{0.1cm}%
\href{ https://zenodo.org/records/10888880}{\includegraphics{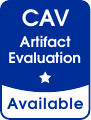}}%
\hspace{9cm}%
\includegraphics{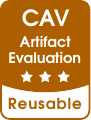}%
}}

\begin{abstract}

Inner-approximate reachability analysis involves calculating subsets of reachable sets, known as inner-approximations. This analysis is crucial in the fields of dynamic systems analysis and control theory as it provides a reliable estimation of the set of states that a system can reach from given initial states at a specific time instant. In this paper, we study the inner-approximate reachability analysis problem based on the set-boundary reachability method for systems modelled by ordinary differential equations, in which the computed inner-approximations are represented with zonotopes. The set-boundary reachability method computes an inner-approximation by excluding states reached from the initial set's boundary. The effectiveness of this method is highly dependent on the efficient extraction of the exact boundary of the initial set. To address this, we propose methods leveraging boundary and tiling matrices that can efficiently extract and refine the exact boundary of the initial set represented by zonotopes. Additionally, we enhance the exclusion strategy by contracting the outer-approximations in a flexible way, which allows for the computation of less conservative inner-approximations. To evaluate the proposed method, we compare it with state-of-the-art methods against a series of benchmarks. The numerical results demonstrate that our method is not only efficient but also accurate in computing inner-approximations.

\end{abstract}

\input{introduction}
\input{related_work}
\input{preliminaries}

\input{method}

\input{experiment}
\input{conclusion}

\bibliography{main}
\bibliographystyle{splncs04}

\input{appendix}

\end{document}

%% file: introduction.tex
\section{Introduction}

Reachability analysis involves the computation of reachable sets, which are sets of states achieved either through trajectories originating in a given initial set (i.e., forward reachable sets) or through the identification of initial states from which a system can reach a specified target set (i.e., backward reachable sets)  \cite{mitchell2007comparing}. This problem is fundamental and finds motivation in various applications such as formal verification, controller synthesis, and estimation of regions of attraction. As a result, it has garnered increasing attention from both industrial and academic communities, leading to the development of numerous  theoretical results and computational approaches \cite{althoff2021set}.  For many systems, exact reachability  analysis is shown to be undecidable  \cite{henzinger1995s}, particularly in the case of nonlinear systems. Hence, approximation methods are often employed. However, in order to use these approximations as a basis for formal reasoning about the system, it is crucial that they possess certain guarantees. Specifically, it is desirable for the computed approximation to either contain or be contained by the true reachable set, resulting in what are known as outer-approximations and inner-approximations.

This paper focuses on inner-approximate reachability analysis, which calculates an inner-approximation of the reachable set for systems described by ordinary differential equations (ODEs). The inner-approximate reachability analysis has various applications. For instance, it can be used to falsify a safety property by performing forward inner-approximate reachability analysis, which computes an inner-approximation of the forward reachable set \cite{li2018simplecar}. If the computed inner-approximation includes states that violate the safety property, then the safety property is not satisfied. On the other hand, it can be used to find a set of initial states that satisfy a desired property by performing backward inner-approximate reachability analysis \cite{lakshmikantham1990practical}. Recently, it has been applied to path-planning problems with collision avoidance \cite{schoels2020nmpc}. Several methods have been proposed for the inner-approximate reachability analysis, such as Taylor models \cite{chen2015reachability}, intervals \cite{goubault2017forward}, and polynomial zonotopes \cite{kochdumper2020computing}.

In the computation of inner-approximations, the accumulation of computational errors, known as the wrapping effect \cite{neumaier1993wrapping}, becomes pronounced with the propagation of the initial set. To overcome this, a common approach is to partition the initial set into smaller subsets, enabling independent computations on each subset. However, this widely used method often results in an excessively large number of subsets, causing burdensome computation. Consequently, in \cite{xue2016under,kochdumper2020computing}, set-boundary reachability methods were developed based on a meticulous examination of the topological structure.
These methods contract a pre-computed outer-approximation by excluding the reachable set from the boundary of the initial set, resulting in an inner-approximation. Compared to the partition of the entire initial set, set-boundary methods alleviate the computational burden and enhance the tightness of results by focusing on splitting only the boundary of the initial set. Hence, the precision of extracting and refining \footnote{If $\mathcal{P}$ and $\mathcal{Q}$ are partitions (or covers) of a set $X$, then $\mathcal{P}$ refines $\mathcal{Q}$ if for every $U\in \mathcal{P}$, there is $V\in \mathcal{Q}$ such that $U\subset V$.} the boundary of initial set significantly influences the non-conservativeness of inner-approximation aimed to compute. However, existing boundary operations have limitations that impact the precision and application of set-boundary reachability methods, either restricting the initial sets to be interval-formed \cite{kochdumper2020computing} or utilizing interval sets to outer-approximate the set boundary \cite{xue2016under}, which leads to an overly conservative inner-approximation and hinders the application of set-boundary reachability methods.

On this concern, this paper proposes a novel set-boundary reachability method focusing on efficient extraction and refinement of the initial set's boundary, along with flexible inner-approximation generations. 
 We adopt zonotopes as the abstract representation of states due to their remarkable advantages: the facets of a zonotope remain zonotopes and can be split into non-overlapping subsets while preserving their zonotopic nature. 
Based on the symmetric property of zonotope's boundary, we propose an algorithm which can efficiently extract all facets of zonotopes.
 To further refine the extracted boundary, a fundamental algorithm is developed to partition a zonotope into smaller, non-overlapping zonotopes, termed tiling algorithm. This algorithm leverages two innovative data structures, named as boundary and tiling matrices, providing a clear and efficient implementation of the partition procedure. Complexity analysis demonstrates the superior advantages of the tiling algorithm in computational complexity compared to the existing method \cite{kabi2020synthesizing}. Finally, we contract a pre-computed outer-approximation of reachable set to obtain an inner-approximation, which is achieved by excluding the outer-approximation of the reachable set from the refined boundary of the initial set. In contrast to proportionally shrinking the shape of computed outer-approximation utilized in existing method \cite{xue2016under}, we provide a more flexible strategy that allows an adaptive modification on the configuration of zonotopic outer-approximations, leading to more non-conservative inner-approximations.

\color{black}
The main contributions of this paper are as follows: 
\color{black}
\begin{itemize}

\item  \textit{A Non-overlapping Zonotope Splitting Algorithm.} We present a novel algorithm that efficiently splits a zonotope into non-overlapping subsets, while preserving their zonotopic properties. By utilizing boundary and tiling matrices, our algorithm offers a more straightforward implementation with improved computational complexity compared to existing methods.

  \item \textit{An Adaptive Contraction Strategy.} We put forward an adaptive contraction strategy for computing a zonotopic inner-approximation of the reachable set. This strategy, compared to existing methods, provides a more flexible approach for the contraction of the pre-computed outer-approximations, generating  less conservative inner-approximations.
     \item  \textit{A Prototype Tool - BdryReach.} We have developed a prototype tool named BdryReach to implement our proposed approach, which is available from \url{https://github.com/ASAG-ISCAS/BdryReach}. Numerous evaluations on various benchmarks demonstrate that BdryReach outperforms state-of-the-art tools in terms of efficiency and accuracy.  


    
\end{itemize}

\color{red}

\color{black}




%% file: related_work.tex
\subsection*{Related Work}
\label{related work}

\subsubsection{Inner-approximation analysis}
The methods for inner-approximation computation are generally categorized into two main groups: constraint solving methods and set-propagation methods. Constraint solving methods avoid the explicit computation of reachable sets, but have to address a set of quantified constraints, which are generally constructed via Lyapunov functions \cite{branicky1998multiple}, occupation measures \cite{korda2013inner} and  equations relaxation \cite{xue2019inner,xue2022reach}. However, solving these quantified constraints is usually computationally intensive (except the case of polynomial constraints for which there exists advanced tools such as semi-definite programming). 

The set propagation method is an extension of traditional numerical methods for solving ODEs using set arithmetic rather than point arithmetic. 
While this method is simple and interesting, a major challenge is the propagation and accumulation of approximation errors over time.
To ease this issue efficiently, various methods employing different representations have been developed. \cite{rwth2014under} presented a Taylor model backward flowpipe method that computes inner-approximations by representing them as the intersection of polynomial inequalities. However,  this approach relied on a computationally expensive interval constraint propagation technique to ensure the validity of the representation. In \cite{goubault2017forward}, an approach is proposed to compute interval inner-approximations of the projection of the reachable set onto the coordinate axes for autonomous nonlinear systems. This method is later extended to systems with uncertain inputs in \cite{goubault2019inner}. However, they cannot compute an inner-approximation of the entire reachable set, as studied in the present work. \cite{xue2016under} proposed a set-boundary reachability method which propagates the initial set's boundary to compute an polytopic inner-approximation of the reachable set.  However, it used computationally expensive interval constraint satisfaction techniques to compute a set of intervals which outer-approximates the initial set's boundary. Recently, inspired by the computational procedure in \cite{xue2016under}, \cite{kochdumper2020computing} introduced  a promising method based on polynomial zonotopes to compute inner-approximations of reachable sets for systems with an initial set in interval form. The method presented in this work is also inspired by the in \cite{xue2016under}. 
However, we propose efficient and accurate algorithms for extracting and refining the boundary of the initial set represented by zonotopes
and an adaptive strategy for contracting outer-approximations, facilitating the computation of non-conservative inner-approximations. 
\subsubsection{Splitting and tiling of zonotopes}
To mitigate wrapping effect \cite{neumaier1993wrapping} and enhance computed results, it is a common way to split a zonotope into smaller zonotopes during computation. 
Despite zonotopes being special convex polytopes with centrally symmetric faces in all dimensions \cite{ziegler2012lectures}, traditional polytope splitting methods such as \cite{bajaj1996splitting,herrmann2008splitting} cannot be directly applied. The results obtained through these approaches are polytopes, not necessarily zonotopes.
In the works \cite{althoff2008reachability,wan2009numerical}, they split a zonotope by bisecting it along one of its generators. However, the sub-zonotopes split by this way often have overlap parts, resulting in loss of precision and heavy computation burden. Hence, there is a pressing need for methods that split a zonotope into non-overlapping sub-zonotopes. The problem of zonotopal tiling, i.e., paving a zonotope by tiles (sub-zonotopes) without gaps and overlaps, is an important topic in combinatorics and topology \cite{bjorner1999oriented,ziegler2012lectures}. In the realm of zonotopal tiling, Bohne-Dress theorem \cite{richter1994zonotopal} plays a crucial role by proving that a tiling of a zonotope can be uniquely represented by a collection of sign vectors or oriented matroid. Inspired by this theorem, \cite{kabi2020synthesizing} developed a tiling method  by enumerating the vertices of the tiles as sign vectors of the so-called hyperplane arrangement \cite{mcmullen1971zonotopes} corresponding to a zonotope. However, in this paper we provide a novel and more accessible method for constructing a zonotopal tiling, which has better computational complexity.

The remainder of this paper is organized as follows. The inner-approximate reachability problem of interest is presented in Sect. \ref{preliminaries}. Then, we elucidate our reachability computational approach in Sect. \ref{method} and evaluate it in Sect. \ref{experiment}. Finally, we summarize the paper in Sect. \ref{conclusion}.

%% file: preliminaries.tex
\section{Preliminaries}
\label{preliminaries}
\subsection{Notation}

The notations and operations concerning space, vectors, matrices, and sets utilized in this paper are presented in Table \ref{notation}, where the symbols and descriptions for operations on vectors, matrices, and sets are mainly illustrated with specific examples of a vector $\bm{x}$, a matrix $\bm{M}$, and a set $\Delta$.

\begin{table}[htbp]
       \caption{Notations utilized in the paper}
    \centering
    \setlength{\tabcolsep}{1.2mm}{
    \begin{tabular}{cccc}
  \toprule
   \textbf{Symbol}&\textbf{Description} &\textbf{Symbol}& \textbf{Description}\\
  \midrule
    $\mathbb{R}^{k}$ &   $k$-dimensional real space & $\mathbb{R}^{m,n}$ & space of $m \times n$ real matrices \\
    $\mathbb{N}_{[m,n]}$ & non-negative integers in $[m,n]$ &$\bm{x}_1\cdot\bm{x}_2$ & inner product of $\bm{x}_1$ and $\bm{x}_2$\\
    $\bm{x},\bm{y},\cdots$ & vectors, boldface lowercase &$\bm{M},\bm{N},\cdots$ &matrices, boldface uppercase  \\
       $\bm{0}$ &vectors with all zero entries  & $\bm{1}$ & vectors with all one entries\\
    $\bm{M}(i, \cdot)$  & $i$-th row vector of  $\bm{M}$&$\bm{M}(\cdot, j)$ &$j$-th column vector of  $\bm{M}$\\
    $\bm{x}(i)$ & $i$-th entry of  $\bm{x}$ & $\bm{M}(i, j)$ & $j$-th entry in $i$-th row of  $\bm{M}$\\
    $\text{rows}(\bm{M})$ & number of rows of  $\bm{M}$&$\text{cols}(\bm{M})$ & number of  columns of  $\bm{M}$\\
    $\bm{M}(-1,\cdot)$ & last row  of $\bm{M}$&$\bm{M}(\cdot,-1)$ &last column of $\bm{M}$\\
    $\bm{M}^{[i]}$ &delete $i$-th row of  $\bm{M}$ & $\bm{M}^{\langle i \rangle}$ & delete $i$-th column of $\bm{M}$\\
     $[\bm{M}; \bm{x}^{\intercal}]$ &add  $\bm{x}$ to  last row  of $\bm{M}$ &$(\bm{M}, \bm{x})$ & add  $\bm{x}$ to  last column  of $\bm{M}$\\
    $\rank(\bm{M})$ & rank of $\bm{M}$&$\Vert \bm{x} \Vert$ & norm of $\bm{x}$\\
    $\Delta^\circ$& interior of set $\Delta$  &$\partial \Delta$& boundary of set $\Delta$ \\
    $\vert\Delta\vert$ & cardinality of set $\Delta$ &$\mathcal{S}_1 \setminus \mathcal{S}_2$ &$\{s~|~s\in \mathcal{S}_1 \wedge s\notin \mathcal{S}_2\}$\\
  \bottomrule
\end{tabular}}
    \label{notation}
\end{table}

\subsection{Problem Statement}

This paper considers nonlinear systems which are modelled by ordinary differential equations of the following form:
\begin{equation}
\label{system}
    \dot{\bm{x}} = \bm{f}(\bm{x}) 
\end{equation}
where $\bm{x} \in \mathbb{R}^n$ and $\bm{f}$ is a locally Lipschitz continuous function. Thus, given an initial state $\bm{x}_0$, there exists an unique solution $\phi(\cdot;\bm{x}_0): [0,T_{\bm{x}_0})\rightarrow \mathbb{R}^n$ to system \eqref{system}, where $[0,T_{\bm{x}_0})$ is the maximal time interval on which  $\phi(\cdot;\bm{x}_0)$ is defined.



Given a set $\mathcal{X}_0$ of initial states, the reachable set is defined as follows:
\begin{definition}[Reachable Set]
\label{Reachable Set}
Given system \eqref{system} and an initial set $\mathcal{X}_0$, the reachable set at time $t>0$ is
\[\Phi(t;\mathcal{X}_0) \triangleq \{ \phi(t;\bm{x}_0) \mid \bm{x}_0 \in \mathcal{X}_0 \}.\]
\end{definition}
The exact reachable set $\Phi(t;\mathcal{X}_0)$ is usually impossible to be computed, especially for nonlinear systems. Outer-approximations and inner-approximations are often computed for formal reasoning on the system. 
\begin{definition}
  Given an initial set $\mathcal{X}_0$ and a time instant $t> 0$, an outer-approximation $O(t;\mathcal{X}_0)$ of the reachable set $\Phi(t;\mathcal{X}_0)$ is a superset of the set $\Phi(t;\mathcal{X}_0)$, i.e., \[\Phi(t;\mathcal{X}_0)\subseteq O(t;\mathcal{X}_0);\]  an inner-approximation $U(t;\mathcal{X}_0)$ of the reachable set $\Phi(t;\mathcal{X}_0)$ is a subset of the set $\Phi(t;\mathcal{X}_0)$, i.e., \[U(t;\mathcal{X}_0) \subseteq \Phi(t;\mathcal{X}_0).\]
\end{definition}

In this paper, we focus on the computation of an inner-approximation represented by zonotopes. Zonotope is a special class of convex polytopes with the centrally symmetric nature. It can be viewed as a Minkowski sum of a finite set of line segments, known as {\textit{G-representation}}, which is defined as the following.
\begin{definition}[Zonotope]
\label{Zonotope}
A zonotope $Z$ with $p$ generators is a set
\begin{equation*}
    \begin{aligned}
Z&=\left\{ \bm{x} \in \mathbb{R}^{n} \Big|  \bm{x}=\bm{c}+\sum\nolimits_{i=1}^{p} \alpha_{i} \cdot \bm{g_i}, -1 \leq \alpha_{i} \leq 1\right\} \\
&=\left\{ \bm{x} \in \mathbb{R}^{n} \Big|  \bm{x}=\bm{c}+\bm{G\alpha}, -\bm{1} \leq \bm{\alpha} \leq \bm{1}\right\},  
    \end{aligned}
\end{equation*}
denoted by $Z=\left \langle \bm{c},\bm{G} \right \rangle$, where $\bm{c}\in \mathbb{R}^n$ is referred as center and $\bm{g}_1, \cdots, \bm{g}_p \in \mathbb{R}^n$ as generators of zonotope. $\bm{G} = (\bm{g}_i)_{1 \leq i\leq p} \in \mathbb{R}^{n,p}$ is called generator matrix. 
\end{definition}


For a zonotope $Z=\left \langle \bm{c},\bm{G} \right \rangle$ in space $\mathbb{R}^n$, 
it is called {\textit{$k$-dimensional}} if $\rank(\bm{G}) = k, k \leq n$. A $k$-dimensional zonotope can be reduced into space $\mathbb{R}^k$ without altering its shape. Furthermore, the facets of a $k$-dimensional ($k\geq 1$) zonotope are $(k-1)$-dimensional zonotopes. 
If an $n$-dimensional zonotope has $n$ independent generators, then it's called {\textit{parallelotope}}. Additionally, If there is a zonotope $Z^\prime$ such that $Z^\prime \subsetneqq Z$, $Z^\prime$ is called a {\textit{sub-zonotope}} of $Z$.

%% file: method.tex
\section{Methodology}
\label{method}
In this section we introduce our set-boundary reachability method to compute inner-approximations of reachable sets. Firstly, the framework of our method is presented in Subsection \ref{framework}. Then, we introduce the algorithm of extracting the exact boundary of a zonotope in Subsection \ref{boundary}, 
the tiling algorithm for boundary refinement in Subsection \ref{refinement} and the strategy for computing an inner-approximation via contracting an outer-approximation in Subsection \ref{contraction}.  

\subsection{Inner-approximation Computation Framework}
\label{framework}

The framework of computing inner-approximations in this paper 
follows the one proposed in \cite{xue2016under}, but with minor modifications. 

Given system \eqref{system} with an initial set $\mathcal{X}_0$, represented by a zonotope, and a time duration $T = Nh$, where $h>0$ is the time step and $N$ is a non-negative integer, we compute a zonotopic inner-approximation $U((k+1)h;\mathcal{X}_0)$  of the reachable set $\Phi((k+1)h;\mathcal{X}_0)$ for $k\in\{0,1,\cdots,N\}$.  The inner-approximation $U_{k+1}=U((k+1)h;\mathcal{X}_0)$ is computed based on $U_{k}=U(kh;\mathcal{X}_0)$ ($U_{0}:= \mathcal{X}_0)$ with the following procedures: 
\begin{enumerate}
    \item extract and refine the boundary $\partial U_{k}$ of $U_{k}$ ;
    \item compute a zonotopic outer-approximation $O(h;U_{k})$ of reachable set $\Phi(h;U_{k})$, and an outer-approximation $O(h;\partial U_{k})$  of  reachable set $\Phi(h;\partial U_{k})$. These outer-approximations can be computed using existing 
    zonotope-based approaches such as \cite{althoff2008reachability}.
    \item contract $O(h;U_k)$ to obtain a zonotopic inner-approximation candidate $U'_{k+1}$ by excluding the set $O(h;\partial U_k)$, i.e., let $U'_{k+1} \cap O(h;\partial U_k) = \emptyset$.
    \item compute an outer-approximation of the reachable set $O(h;\bm{c})$ of the time-inverted system $\dot{\bm{x}}=-\bm{f}(\bm{x})$ with the single initial state $\bm{c}$, where $\bm{c}$ is the center of the zonotope $U'_{k+1}$. If the computed outer-approximation $O(h;\bm{c})$ is included in the set $U_k$, then $U_{k+1}:=U'_{k+1}$ is an inner-approximation of  the reachable set $\Phi((k+1)h;\mathcal{X}_0)$.
\end{enumerate}

\begin{figure}[htbp]
    \centering
        \includegraphics[width=4in]{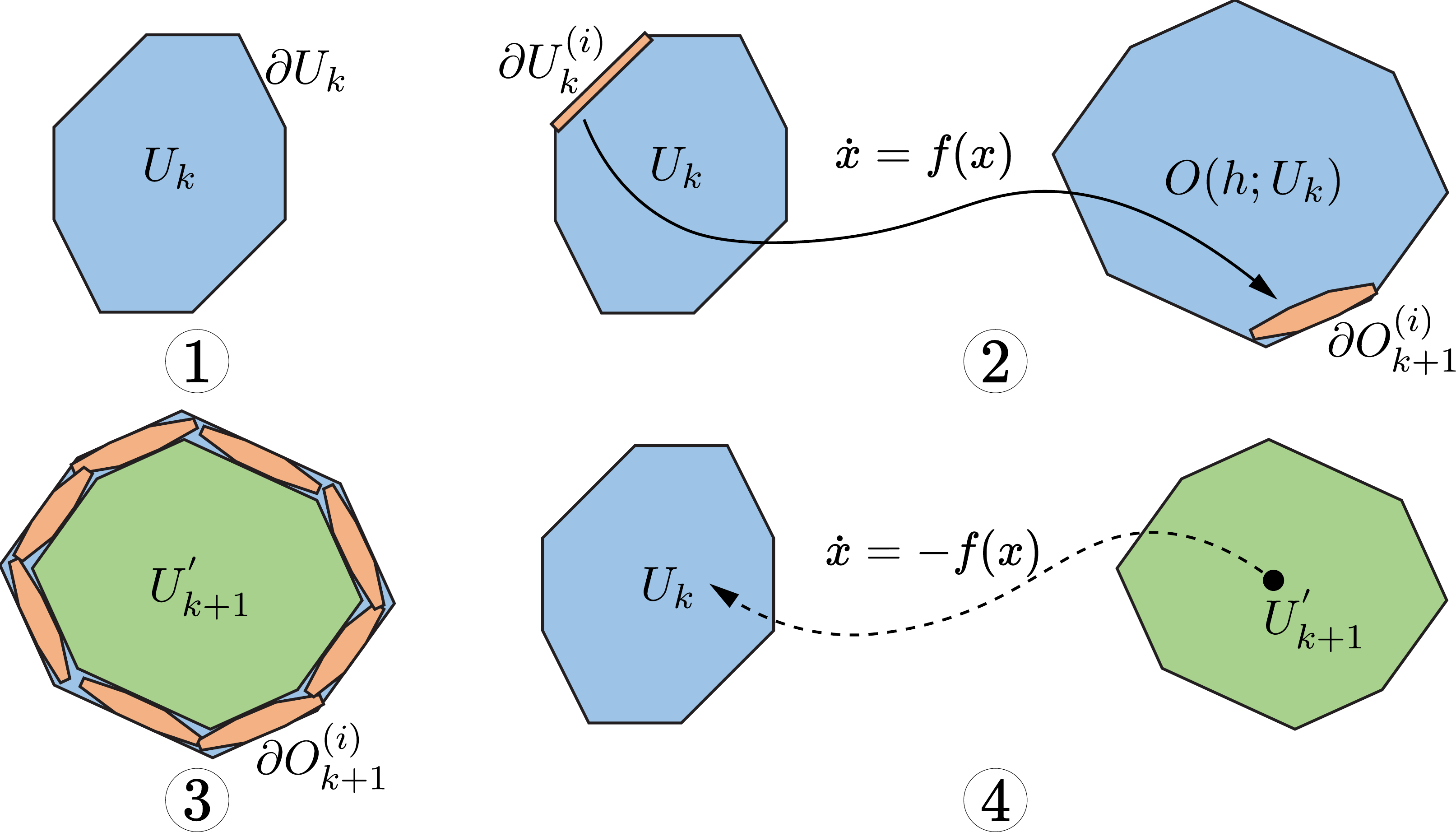}
    \caption{
    Illustration of inner-approximation computation framework}
    \label{fig:framework}
 
\end{figure}

The overall computational workflow is visualized in Fig. \ref{fig:framework}. There are three computational procedures that affect the efficacy (i.e., accuracy and efficiency) of  inner-approximation computation in the aforementioned framework: the extraction and refinement of the boundary $\partial U_k$, reachability analysis for computing outer-approximations $O(h;U_k), O(h;\partial U_k)$, and contraction of $O(h;U_k)$ to obtain an inner-approximation candidate $U_{k+1}^\prime$. Since there are well-developed reachability algorithms in existing literature for computing outer-approximations such as \cite{girard2005reachability,althoff2008reachability}, we in the following focus on other two computational procedures. 
For the first one, as the outer-approximation computed $O(h;\partial U_k)$ would be excluded from $O(h;U_k)$, the accuracy of $O(h;\partial U_k)$ 
significantly affects the one of $U_{k+1}$. Additionally, the accuracy of  $O(h;\partial U_k)$ strongly correlates with the size of $\partial U_k$. To improve the accuracy of $U_{k+1}$, two algorithms are proposed: one for extracting and the other for tiling the boundary of a zonotope (i.e., splitting the boundary into sub-zonotopes without overlaps). As for the third one, an adaptive strategy is developed to make the inner-approximation $U_{k+1}$ much tighter. This is achieved by contracting $O(h;U_k)$ in a flexible way, deviating from the proportional reduction of the size of $O(h;U_k)$ in the existing methods \cite{xue2016under}.
\subsection{Extraction of Zonotopes' Boundaries}
\label{boundary}

In this subsection we introduce the algorithm for extracting the exact boundary of a zonotope. The concept of cross product of a matrix provided by \cite{mortari1997n} will be utilized herein, which is formulated below.
\begin{definition}[Cross Product]
\label{def:cross product}
Given a matrix $\bm{M} \in \mathbb{R}^{n,n-1}$ in which the column vectors are linearly independent.  The cross product of $\bm{M}$ is a vector of the following form:
$$\mathtt{CP}(\bm{M}) = \left( \operatorname{det}\left(\bm{M}^{[1]}\right),
\cdots,(-1)^{i+1} \operatorname{det}\left(\bm{M}^{[i]}\right),  \cdots,(-1)^{n+1} \operatorname{det}\left(\bm{M}^{[n]}\right)
\right)^\intercal,$$
where $\operatorname{det}(\cdot)$ is the determinant of a matrix. 
\end{definition}
The cross product of $\bm{M} \in \mathbb{R}^{n,n-1}$ can be viewed as the normal vector of the hyperplane spanned by $n-1$ linearly independent column vectors in $\bm{M}$.

\begin{figure}[htbp]
    \centering
    \includegraphics[width = 2.1in]{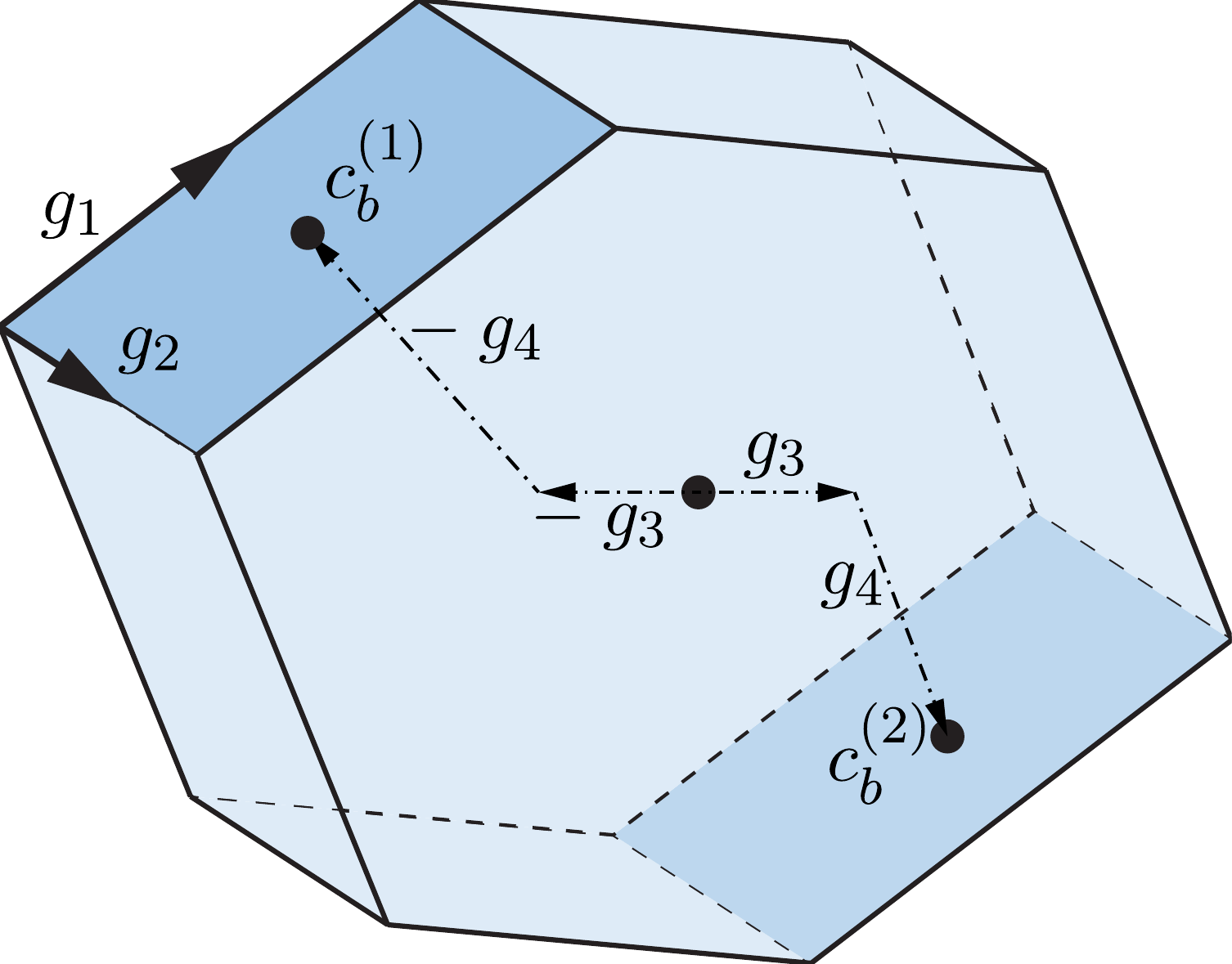}
    \caption{Illustration of boundary extraction algorithm}
    \label{fig:boundary extraction algorithm}
\end{figure}

The boundary extraction algorithm is established on the fact that  a zonotope is centrally symmetric and each facet, which is a zonotope, has an congruent facet on the opposite side of the center (e.g., two dark blue facets in Fig.  \ref{fig:boundary extraction algorithm}).


Given an $n$-dimensional zonotope $Z=\left \langle \bm{c},\bm{G} \right \rangle$, where $\bm{c}\in \mathbb{R}^n$ and $\bm{G} = (\bm{g}_i)_{1 \leq i\leq p} \in \mathbb{R}^{n,p}$, for each two symmetric facets, they lie in parallel hyperplanes and share the same generators. The two parallel hyperplanes are spanned by a part of generators of $Z$, which can form a submatrix of $\bm{G}$ with rank $n-1$. In boundary extraction algorithm, i.e., Alg. \ref{alg: boundaries of a zonotope}, we  firstly enumerate all potential $n \times (n-1)$ submatrices of $\bm{G}$ which are able to span a hyperplane. For a certain hyperplane spanned by a submatrix $\bm{B}_b$, to confirm the center and generators of its corresponding facets, we compute its normal vector by the cross product operator $\mathtt{CP}(\cdot)$, then the center of the two symmetric facets can be respectively determined by moving the center $\bm{c}$ along the positive and negative directions of generators which are not perpendicular to $\mathtt{CP}(\bm{B}_b)$, and the generator matrix of these corresponding facets can be represented by $\bm{B}_b$ appending generators parallel to the hyperplane. The visible operations stated above are shown in Fig. \ref{fig:boundary extraction algorithm}.

\begin{algorithm}[tp]

\caption{Boundary Extraction Algorithm}
\label{alg: boundaries of a zonotope}
\textbf{Input}: An $n$-dimensional zonotope $Z=\left \langle \bm{c},\bm{G} \right \rangle, \bm{G} = (\bm{g}_i)_{1 \leq i\leq p} \in \mathbb{R}^{n \times p}$.\\
\textbf{Output}: The boundary of the zonotope $Z$, i.e., $\partial Z$.
\begin{algorithmic}[1] 

\STATE $\partial Z:= \emptyset$

\STATE $\mathcal{B}:= \{\bm{B}_b =(\bm{g}_i)_{i \in \{k_1,\cdots, k_{n-1}\}}, 1 \leq k_1 < \cdots < k_{n-1} \leq p\mid \rank(\bm{B}_b) = n-1\}$
\WHILE{ $\mathcal{B} \neq \emptyset$}
\STATE{$\bm{v} := \mathtt{CP}(\bm{B}_b)$}
\STATE $\bm{B} := \bm{B}_b = (\bm{g}_{k_1}, \bm{g}_{k_2}, \cdots, \bm{g}_{k_{n-1}}) \in \mathcal{B}$, $\bm{c}_b^{(i)}  := \bm{c}, i \in \{1,2\}$



\FORALL{$\bm{g}_k=G(\cdot,k), k \in \{1,2,\cdots,p\} \setminus \{k_1,k_2,\cdots,k_{n-1}\}$}
\IF{$\bm{v} \cdot \bm{g}_k = 0$}
\STATE $\bm{B}:= (\bm{B}, \bm{g}_k)$
\ELSIF{$\bm{v}\cdot \bm{g}_k > 0$}
\STATE$\bm{c}_{b}^{(i)} := \bm{c}_{b}^{(i)} + (-1)^{i}\bm{g}_k$, $i\in\{1,2\}$
\ELSE 
\STATE $\bm{c}_{b}^{(i)}:= \bm{c}_{b}^{(i)} - (-1)^{i}\bm{g}_k$, $i\in\{1,2\}$
\ENDIF

\STATE $Z_b^{(i)}: = \left \langle \bm{c}_{b}^{(i)}, \bm{B} \right \rangle$, $i\in\{1,2\}$
\STATE $\partial Z:=\partial Z \cup \{Z_b^{(1)},Z_b^{(2)}\}$
\ENDFOR

\FORALL{$\bm{B}_b \in \mathcal{B}$}

\IF{$\bm{B}_b$ is a submatrix of $\bm{B}$}
\STATE $\mathcal{B}: = \mathcal{B} \setminus \{\bm{B}_b\}$
\ENDIF
\ENDFOR
\ENDWHILE

\STATE \textbf{return}  $\partial Z$

\end{algorithmic}

\end{algorithm}

The computation of a zonotope's boundary is summarized in Alg. \ref{alg: boundaries of a zonotope}. Its soundness, i.e., the set computed by Alg. \ref{alg: boundaries of a zonotope} is equal to the boundary $\partial Z$ of the zonotope $Z$,  is justified in Theorem \ref{thm: Boundary of a Zonotope}, whose proof is available in Appendix \ref{app: proof}.
In order to enhance the understanding of Alg. \ref{alg: boundaries of a zonotope}, we provide a simple example, Example \ref{ex:boundary} in Appendix \ref{app:extra examples}, to illustrate the computational process of Alg. \ref{alg: boundaries of a zonotope}.

\begin{remark}

    In space $\mathbb{R}^n$, if  a zonotope $Z=\left \langle \bm{c},\bm{G} \right \rangle$ isn't $n$-dimensional, i.e., $\rank(\bm{G}) < n$, then the boundary of this zonotope is itself. 
\end{remark}




\color{black}


\color{black}

\begin{theorem}[Soundness of boundary extraction algorithm]
\label{thm: Boundary of a Zonotope}
Given an $n$-dimensional zonotope $Z=\left \langle \bm{c},\bm{G} \right \rangle$ with $p$ generators, the set computed by Alg. \ref{alg: boundaries of a zonotope} is equal to its boundary $\partial Z$.
\end{theorem}

\subsubsection{The complexity of boundary extraction algorithm}
For an $n$-dimensional zonotope $Z=\left \langle \bm{c},\bm{G} \right \rangle, \bm{G} \in \mathbb{R}^{n, p}$, it has $M$ facets, where $M \leq \binom{p}{n-1}$.
The number of $n\times (n-1)$ submatrices of $\bm{G}$ is $\binom{p}{n-1}$, and the computation of the rank of an $n \times (n-1)$ matrix has the complexity $O(n(n-1)^2)$ (using QR decomposition), then the computation in Line 2 has the complexity $O(n(n-1)^2\binom{p}{n-1})$. In ``while'' Loop (Line 3-22), it has $\frac{M}{2}$ iterations. 
For the operation $\mathtt{CP}(\cdot)$ on an $n \times (n-1)$ matrix, its complexity is $n\DET(n-1)$, where $\DET(n)$ denote the complexity of computing a determinant of an $n\times n$ square matrix. By LU-decomposition, $\DET(n)$ is $O(n^3)$, however by Coppersmith–Winograd algorithm \cite{fisikopoulos2016faster}, it can reach $O(n^{2.373})$. For each $\bm{B}_b \in \mathcal{B}$, checking the inner product between $\bm{v}$ and remaining generators has $p-n+1$ loops, and the inner product has complexity $O(n)$. Thus, the complexity of Alg. \ref{alg: boundaries of a zonotope} is $\frac{M}{2}\left(n\DET(n-1)+n(p-n+1)\right)+O(n(n-1)^2\binom{p}{n-1}) = O\left(Mn(\DET(n-1)+p)+n(n-1)^2\binom{p}{n-1}\right)$.

\subsection{Zonotopal Tiling and Boundary Refinement}
\label{refinement}




\color{black}
This subsection introduces our tiling algorithm which can split a zonotope into sub-zonotopes without overlaps and then elaborates how this tiling algorithm is employed to refine the boundaries of zonotopes.

 The boundary matrix, which is constructed according to Alg. \ref{alg: boundaries of a zonotope}, plays an important role in our tiling algorithm. Its entries are able to characterize the centers and generators for all facets of a zonotope.  
\begin{definition}[Boundary Matrix]
\label{def boundary matrix}
Given an $n$-dimensional zonotope $Z=\left \langle \bm{c},\bm{G} \right \rangle$ with $M$ facets, where $\bm{c} \in \mathbb{R}^n$ and $\bm{G} \in \mathbb{R}^{n,p}$, its boundary matrix $\bm{B} \in \mathbb{R}^{M,p}$ is a matrix whose each entry is 0,1,or -1, where 
\begin{itemize}
    \item[1.] $\bm{B}(i,j)=0$ implies that the $j$-th generator $\bm{g}_j$ is a  generator of the $i$-th facet (corresponding to Line 8 in Alg. \ref{alg: boundaries of a zonotope});
    \item[2.] $\bm{B}(i,j)=-1$ implies that in order to obtain the center of  
    the $i$-th facet, the MINUS operator is applied to the $j$-th generator $\bm{g}_j$ (corresponding to Line 10 and 12 in Alg. \ref{alg: boundaries of a zonotope}); 
    \item[3.] $\bm{B}(i,j)=1$ implies that in order to obtain the center of  
    the $i$-th facet, the PLUS operator is applied to the $j$-th generator $\bm{g}_j$  (corresponding to Line 10 and 12 in Alg. \ref{alg: boundaries of a zonotope}).
\end{itemize}
\end{definition}




From the boundary matrix of a zonotope, one can obtain all its facets. Appendix \ref{app:extra examples} provides an example (Example \ref{ex:boudary matrix}) to illustrate this claim.



{


Another matrix, tiling matrix, is constructed to store the outcomes of the tiling algorithm, i.e., all the non-overlapping sub-zonotopes whose union covers the original zonotope. Similar to the boundary matrix, a row of tiling matrix represents a sub-zonotope.
\begin{definition}[Tiling Matrix]
\label{def:partition matrix}
Given an $n$-dimensional zonotope $Z=\left \langle \bm{c},\bm{G} \right \rangle$, where $\bm{c} \in \mathbb{R}^n$ and $\bm{G} \in \mathbb{R}^{n,p}$, its tiling matrix $\bm{T} \in \mathbb{R}^{s,p}$ is a matrix satisfying the following conditions:
\begin{enumerate}
    \item its each entry is 0,1 or -1, which has the same meaning with the one in the boundary matrix;
    \item each row defines a sub-zonotope $Z_i$ such that  $\bigcup_{i=1}^s   Z_i=Z$ and $Z^{\circ}_i\cap Z^{\circ}_j=\emptyset$ for $i\neq j$.
\end{enumerate}
\end{definition}

Our tiling algorithm is based an intuitive observation: for a zonotope, moving its one-sided facets towards to the opposite side along the direction of a generator results in a new zonotope with this generator removed, simultaneously, several sub-zonotopes are generated by adding this generator to all these facets. This process, which is visualized in Fig. \ref{fig:tiling of a zonotope}, can be iteratively conducted, until a parallelotope remains. At this point, the tiling algorithm terminates, yielding a collection of tiles denoted as zonotopes that tile the original zonotope. 

\begin{figure}[htbp]
    \centering
    \includegraphics[width=0.9\linewidth]{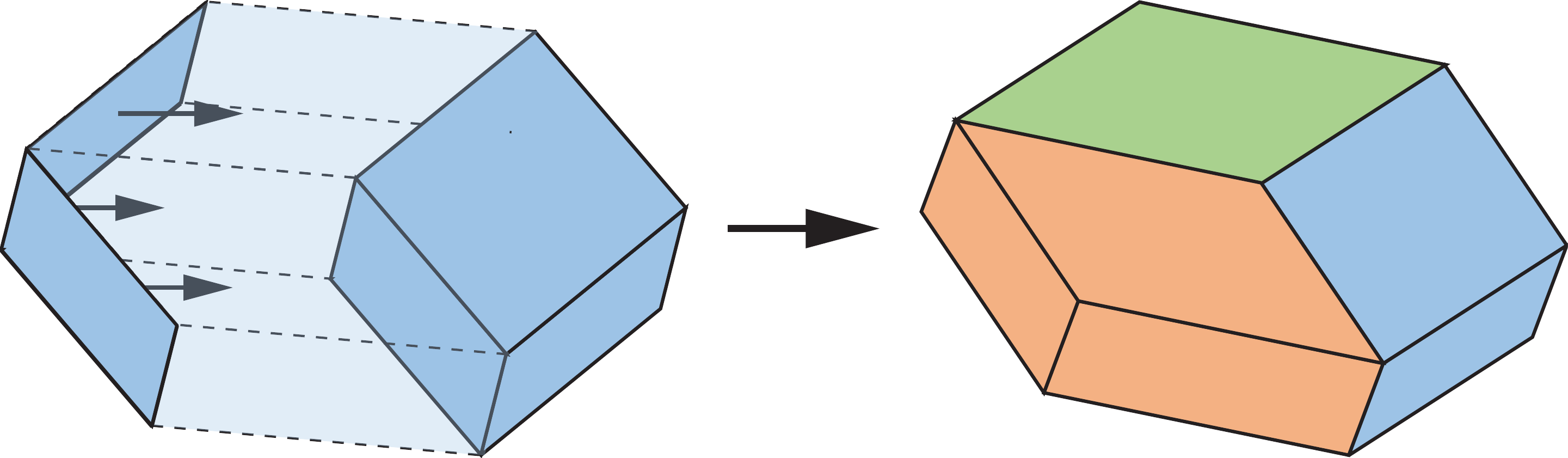}
    \caption{Illustration of one-step tiling}
    \label{fig:tiling of a zonotope}
\end{figure}
The tiling algorithm leverages operations on boundary matrix $\bm{B}$ to implement the facets' movement and sub-zonotopes generation aforementioned. The results of each step, namely the sub-zonotopes after one-step tiling, are recorded in the tiling matrix $\bm{T}$.
 
Given an $n$-dimensional zonotope $Z=\left \langle \bm{c},\bm{G} \right \rangle,$ where $\bm{G} = (\bm{g}_i)_{1 \leq i\leq p} \in \mathbb{R}^{n,p}$, we require that the right-most $n\times n$ submatrix of $\bm{G}$ is full rank to ensure that the sub-zonotopes with one generator removed after one-step tiling remain $n$-dimensional. 
For the specific  $j$-th column of boundary matrix $\bm{B}$, where $ 1\leq j\leq (p-n)$, we process the following operations to its entries:
\begin{enumerate}
    \item if there exist $i$'s such that $\bm{B}(i,j)=-1$, we add these rows in the boundary matrix $\bm{B}$  into the tiling matrix $\bm{T}$ as new rows, but change their $j$-th entry to $0$ in the tiling matrix $\bm{T}$. Meanwhile, the $j$-th entries of these rows in the boundary matrix $\bm{B}$ are modified into $1$, i.e., $\bm{B}(i,j)=1$; 
    \item if there exist $i$'s such that $\bm{B}(i,j)=0$, we delete these rows from the boundary matrix $\bm{B}$.
\end{enumerate}
After the $j$-th iteration, the updated boundary matrix $\bm{B}$ characterizes the boundary of a new zonotope. This new zonotope is derived by removing the first through the $j$-th generators from the original zonotope $Z$. Simultaneously, the sub-zonotopes generated by adding the generator $\bm{g}_j$ to the facets are incorporated into the tiling matrix $\bm{T}$. Finally, after $p-n$ iterations, there remains one parallelotope, whose generator matrix is the right-most $n\times n$ submatrix of $\bm{G}$, we put this parallelotope into the tiling matrix $\bm{T}$ and then output the result. 

The above computational procedures are summarized in Alg. \ref{alg: partition of a zonotope}. Its soundness is justified by Theorem \ref{thm: partition of a Zonotope}, whose proof is available in Appendix \ref{app: proof}.
Moreover, Appendix \ref{app:extra examples} supplements an example (Example \ref{ex tiling}) to illustrate the main steps tiling a zonotope using Alg. \ref{alg: partition of a zonotope}.

\begin{remark}
    For an $n$-dimensional parallelotope, Alg. \ref{alg: partition of a zonotope} only return itself since there is no generator to remove while keeping it $n$-dimensional. However, one can use some simple methods to tile it such as parallelepiped grid. 
\end{remark}

\begin{algorithm}[htbp]
\caption{Tiling Algorithm
}
\label{alg: partition of a zonotope}
\textbf{Input}: An $n$-dimensional zonotope $Z=\left \langle \bm{c},\bm{G} \right \rangle, \bm{G} = (\bm{g}_i)_{1 \leq i\leq p} \in \mathbb{R}^{n, p}$, the right-most $n\times n$ submatrix of $\bm{G}$ is full rank, i.e., $\rank((\bm{g}_i)_{i\in \mathbb{N}_{[p-n+1,p]}}) = n$.\\
\textbf{Output}: Tiling matrix $\bm{T}$.
\begin{algorithmic}[1] 


\STATE Call Alg. \ref{alg: boundaries of a zonotope} to get boundary matrix $\bm{B}$
\STATE $\bm{T}:=[~]$
 
\FOR{$j = 1$ \TO $p-n$}
  
  \FOR{$i = 1$ \TO $\text{rows}(\bm{B})$}

    \IF{$B(i,j) = 0$}
    
        \STATE  
        $\bm{B} :=\bm{B}^{[i]}$
        
    \ENDIF
    
    \IF{$B(i,j) = -1$}
    
        \STATE $\bm{v}^\intercal$ := $B(i,\cdot)$
        \STATE $\bm{v}(j) := 0$ 
        \STATE $\bm{T}:=[\bm{T};\bm{v}^\intercal]$
        \STATE $B(i,j) := 1$
        
    \ENDIF
    
  \ENDFOR
  
\ENDFOR

\STATE $\bm{v}^\intercal = \bm{B}(-1,\cdot)$

\FOR{$j=p-n+1$ \TO $p$}

\STATE $\bm{v}(j) := 0$

\ENDFOR

\STATE $\bm{T}:=[\bm{T};\bm{v}^\intercal]$

\STATE \textbf{return} $\bm{T}$
\end{algorithmic} 
\end{algorithm}}

\begin{remark}
    The sub-zonotopes obtained by Alg. \ref{alg: partition of a zonotope} aren't necessarily parallelotopes. To make the results of tiling are exclusively paralletopes, one can recursive applying Alg. \ref{alg: partition of a zonotope} on each sub-zonotope in tiling matrix $\bm{T}$ until each sub-zonotope has $n$ generators. Additionally, Alg. \ref{alg: partition of a zonotope} allows terminating at any iteration, and the result of each iteration can serve as a tiling of the original zonotope. This flexibility is valuable for controlling the number of partitioned sub-zonotopes. Therefore, our proposed tiling algorithm is particularly well-suited for the inner-approximation computation scenario outlined in this paper, it enables a balance between the computational burden and precision of evaluating $O(h;\partial U_k)$ by constraining the number of sub-zonotopes in the tiling.
\end{remark}

\begin{theorem}[Soundness of tiling algorithm]
\label{thm: partition of a Zonotope}
Given an $n$-dimensional zonotope $Z=\left \langle \bm{c},\bm{G} \right \rangle$ with $p$ generators, the tiling matrix $\bm{T}$ obtained by Alg. \ref{alg: partition of a zonotope} satisfies the conditions in Def. \ref{def:partition matrix}.
\end{theorem}

\subsubsection{The complexity of tiling algorithm}
{ For an $n$-dimensional zonotope  $Z=\left \langle \bm{c},\bm{G} \right \rangle, \bm{G} = (\bm{g}_i)_{1 \leq i\leq p} \in \mathbb{R}^{n,p}$, where $\rank((\bm{g}_i)_{i\in \mathbb{N}_{[p-n+1,p]}}) = n$, 
assume $Z$ has $M$ facets. The calling of Alg. \ref{alg: boundaries of a zonotope} is $O\left(Mn(\DET(n-1)+p)+n(n-1)^2\binom{p}{n-1}\right)$. The size of boundary matrix $\bm{B}$ is $M\times p$, the two-layer ``for'' Loop (Line 3-15) has iterations less than $M(p-n)$, thus the calling of Alg. \ref{alg: boundaries of a zonotope} is dominant in the complexity of tiling algorithm. Consequently, the complexity of Alg. \ref{alg: partition of a zonotope} is $O\left(Mn(\DET(n-1)+p)+n(n-1)^2\binom{p}{n-1}\right)$. }

\subsubsection{Complexity comparison}
Here we  compare the complexity of  tiling algorithm proposed in \cite{kabi2020synthesizing} with ours. For an $n$-dimensional zonotope $Z=\left \langle \bm{c},\bm{G} \right \rangle, \bm{G} \in \mathbb{R}^{n,p}$, assume $Z$ has $M$ facets and $N$ vertexes. 
The main computation procure of algorithm in \cite{kabi2020synthesizing} is computing $\Sigma$ (a set of sign vectors of cells of the arrangement), which is equivalent to enumerate the sign vectors of all the vertexes of $Z$. The computation of $\Sigma$, utilizing a reverse search algorithm \cite{ferrez2005solving}, owns the complexity of $O(np \LP(p,n) \vert\Sigma\vert) = O(Nnp \LP(p,n))$ (the number of sign vectors in $\Sigma$ is equal to $N$), where $\LP(p,n)$ is the time to solve a linear programming (LP) with $p$ inequalities in $n$ variables. There are various algorithms for solving LPs including simplex algorithm, interior point method and their variants. The state-of-the-art algorithms for solving LPs take complexity around $O(n^{2.37})$ \cite{cohen2021solving}. As for the complexity of our algorithm, we have clarified that dominant part in complexity is the procure extracting the boundary of a zonotope, which has the complexity $O\left(Mn(\DET(n-1)+p)+n(n-1)^2\binom{p}{n-1}\right)$. Additionally for zonotope $Z$, the number of its vertexes $N$ is usually much larger than the one of its facets $M$, particularly in high dimension (for example, a hypercube in $\mathbb{R}^n$ has $2^n$ vertexes and $2n$ facets). According to the analysis above, we can conclude the complexity of our tiling algorithm is better than the one of algorithm ($O(Nnp \LP(p,n))$) in \cite{kabi2020synthesizing}.

\begin{figure}[htbp]
    \centering
    \includegraphics[width = \linewidth]{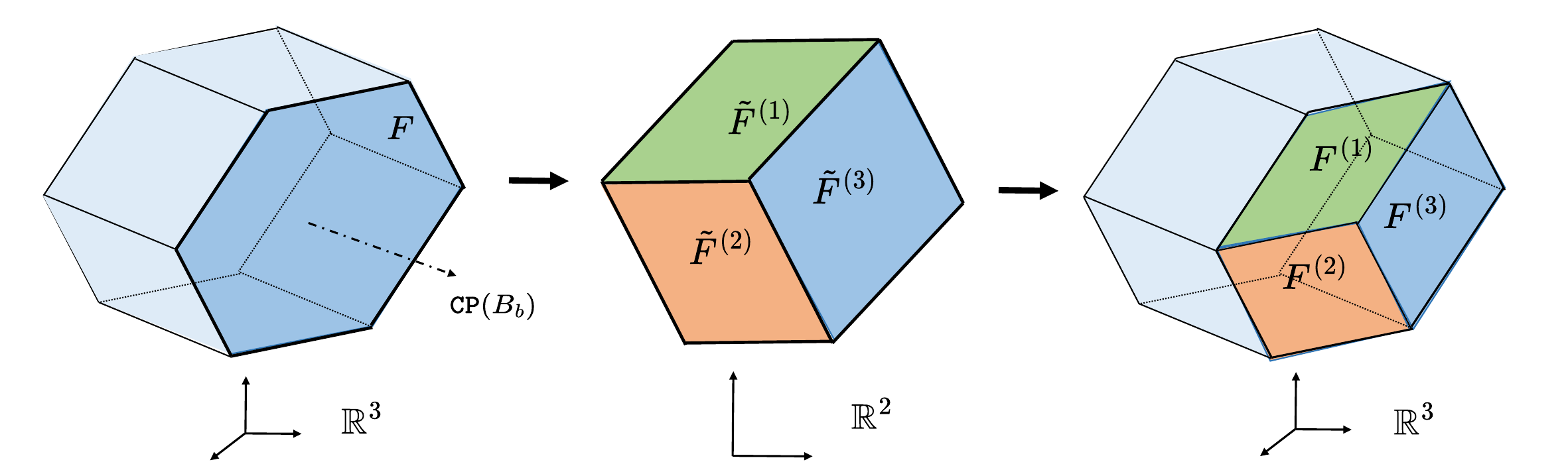}
    \caption{Illustration of boundary refinement}
    \label{fig:refine illustration}
\end{figure}
\subsubsection{Boundary refinement via tiling algorithm} Given an $n$-dimensional zonotope  $Z=\left \langle \bm{c},\bm{G} \right \rangle, \bm{G}\in \mathbb{R}^{n, p}$, for one of its facet $F = \left \langle \bm{c}_b,\bm{G}_b \right \rangle$, we transform it into the space $\mathbb{R}^{n-1}$ with transformation matrix $\bm{B}_b^\intercal$, where $\bm{B}_b$ is the $n\times (n-1)$ submatrix of $\bm{G}_b$ with rank $n-1$,
then the $(n-1)$-dimensional transformed zonotope can be denoted as $\tilde{F}=\left \langle \bm{B}_b^\intercal \bm{c}_b,\bm{B}_b^\intercal \bm{G}_b \right \rangle$. Using tiling algorithm, $\tilde{F}$ can be split into some smaller sub-zonotopes $\{\tilde{F}^{(1)}, \tilde{F}^{(2)}, \cdots, \tilde{F}^{(M)}\}$. For each of $(n-1)$-dimensional sub-zonotopes such as $\tilde{F}^{(1)} = \left \langle \tilde{\bm{c}}_{(1)},\tilde{\bm{G}}_{(1)} \right \rangle$, an inverse transformation recovers it to the zonotope in the space $\mathbb{R}^n$, i.e., $F^{(1)}=\left \langle \bm{c}_{(1)},\bm{G}_{(1)} \right \rangle$, where $\bm{c}_{(1)} = [\bm{B}_b^\intercal; \mathtt{CP}(\bm{B}_b)^\intercal]^{-1}[\tilde{\bm{c}}_{(1)};\mathtt{CP}(\bm{B}_b) \cdot \bm{c}_b]$, $\bm{G}_{(1)} = [\bm{B}_b^\intercal;\mathtt{CP}(\bm{B}_b)^\intercal]^{-1}[\tilde{\bm{G}}_{(1)};\bm{0}^\intercal]$. The main steps of boundary refinement are visualized in Fig. \ref{fig:refine illustration}.

\subsection{Contracting Computed Outer-approximation}
\label{contraction}
In this subsection we present our contraction method, yielding the inner approximation candidate $U_{k+1}^\prime$ by contracting $O(h; U_k)$.
In contrast to the approaches in \cite{xue2016under}, which contracts $O(h; U_k)$ by reducing size proportionally, our contraction method offers a more flexible way. Specifically, the length of each generator of $O(h; U_k)$ can be adjusted and some generators can be removed. The incorporation of this adaptive contraction method enhances the tightness of the computed inner-approximation.

By extracting and refining of boundary $\partial U_{k}$ of $U_{k}$, we get a collection of sub-zonotopes, i.e., $\{\partial U_k^{(i)}\}_{i \in \mathbb{N}_{[1,s]}}$, where $\bigcup_{i=1}^s \partial U_k^{(i)} = \partial U_k$. Then $O(h;\partial U_k)$ can be obtained by uniting all the out-approximations $\partial O_{k+1}^{(i)} := O(h;\partial U_k^{(i)}), i \in \mathbb{N}_{[1,s]}$, i.e., $O(h;\partial U_k) = \bigcup_{i=1}^s \partial O_{k+1}^{(i)}$.

Noticing that the shape of every outer-approximation $\partial O_{k+1}^{(i)}=\left \langle \bm{c}_o,\bm{G}_o \right \rangle$ is usually long and narrow (refer to Fig. \ref{fig:framework}), we choose the top $n-1$ independent generators by norm (such as Euclidean norm) to span a hyperplane, which can be seen as an $(n-1)$-dimensional form approximating $\partial O_{k+1}^{(i)}$. Then we compute the cross product $\mathtt{CP}(\cdot)$ of this hyperplane as its normal vector  to represent the attitude of $\partial O_{k+1}^{(i)}$, denoted by $\mathtt{AT}\left(\partial O_{k+1}^{(i)}\right)$ (i.e. $\mathtt{CP}(\hat{\bm{G}}_o)$, where $\hat{\bm{G}}_o$ contains top $n-1$ indenpent generators of $\bm{G}_o$ by norm).

Initially, we set the inner-approximation candidate $U_{k+1}^\prime := O(h;U_k)$. Subsequently, we iteratively reduce the length of generators and adjust the position (by changing the center) of $U_{k+1}^\prime$ until the intersections between $U_{k+1}^\prime$ and all outer-approximations $\partial O_{k+1}^{(i)}$ become empty sets. 
For each outer-approximation $\partial O_{k+1}^{(i)}$, we begin by shortening the length of generators that are most likely to yield collisions  between $U_{k+1}^\prime$ and $\partial O_{k+1}^{(i)}$, which would prevent the unnecessary contraction of $U_{k+1}^\prime$ and make the result tighter. Heuristically, the generators with directions closest to $\mathtt{AT}\left(\partial O_{k+1}^{(i)}\right)$, or in other words, those most likely to ``perpendicular'' to $\partial O_{k+1}^{(i)}$ (precisely, perpendicular to hyperplane spanned by column vectors of $\hat{\bm{G}}_o$) should be given priority considerations. When encountering a generator that dose not need to be shortened, indicating that $U_{k+1}^\prime$ and $\partial O_{k+1}^{(i)}$ have no overlapping parts, we turn to the next outer-approximation $\partial O_{k+1}^{(i+1)}$.
The details of contraction method proposed is summarized below.
\begin{enumerate}

    \item Initialize inner-approximation candidate $U_{k+1}^\prime :=\left \langle \bm{c}_u,\bm{G}_u \right \rangle= O(h;U_k)$.
    \item For every boundary outer-approximations $\partial O_{k+1}^{(i)}=\left \langle \bm{c}_o,\bm{G}_o \right \rangle, i \in \mathbb{N}_{[1,s]}$, carry out the following processing steps.
    \begin{itemize}
        \item[2a.] Sort the generators $\{\bm{g}_l\}_{1\leq l\leq \cols(\bm{G}_u)}$ of $U_{k+1}^\prime$ according the angle with $\mathtt{AT}\left(\partial O_{k+1}^{(i)}\right)$ from small to large (i.e., $\Vert {\rm cos} \theta \Vert = \frac{\Vert\bm{g}_l\cdot \mathtt{AT}\left(\partial O_{k+1}^{(i)}\right)\Vert}{\Vert \bm{g}_l \Vert \Vert \mathtt{AT}\left(\partial O_{k+1}^{(i)}\right) \Vert}$ from large to small).
        \item[2b.] Loop all the generators according to the sorted order, for the generator $\bm{g}_l$, compute its domain $[\underline{\alpha_l}, \overline{\alpha_l}]$ which intersects with $\partial O_{k+1}^{(i)}$ by LPs \eqref{min} and \eqref{max} (using approach in  \cite[Chapter $4.2.5$]{jaulin2001interval}), where $\bm{\alpha} = (\alpha_1, \cdots, \alpha_l, \cdots, \alpha_{\cols(\bm{G}_u)})^\intercal$, $\bm{\beta} = (\beta_1, \cdots,  \beta_{\cols(\bm{G}_o)})^\intercal$. 
        
        \begin{equation}
        \label{min}
            \begin{aligned}
                 \min\quad & \alpha_l  \\
                 s.t.\quad &\bm{c}_u + \bm{G}_u \bm{\alpha} = \bm{c}_o + \bm{G}_o \bm{\beta}\\
                &-\bm{1} \leq \bm{\alpha}\leq \bm{1}, -\bm{1} \leq \bm{\beta} \leq \bm{1}
            \end{aligned}
        \end{equation}
        \begin{equation}
        \label{max}
            \begin{aligned}
                 \max\quad & \alpha_l  \\
                 s.t.\quad &\bm{c}_u + \bm{G}_u \bm{\alpha} = \bm{c}_o + \bm{G}_o \bm{\beta}\\
                &-\bm{1} \leq \bm{\alpha}\leq \bm{1}, -\bm{1} \leq \bm{\beta} \leq \bm{1}
            \end{aligned}
        \end{equation}    
        When the optimal value of \eqref{min} or \eqref{max} can't be found, then terminate this loop and continue for the next  boundary outer-approximation $\partial O_{k+1}^{(i+1)}$.
        
        \item[2c.] If $[\underline{\alpha_l}, \overline{\alpha_l}]=[-1,1]$, then delete $\bm{g}_l$ from generator matrix $\bm{G}_u$. Else, update the range of $a_l \in \max\{[-1,\underline{a_l} - \epsilon],[\overline{a_l} +\epsilon,1]\} \triangleq [\underline{\gamma}, \overline{\gamma}]$, where the operation $\max\{\cdot,\cdot\}$ means choosing the interval with maximum length and $\epsilon$ is a user-defined small positive number. 
        \item[2d.] Update $\bm{c}_u := \bm{c}_u + 0.5 (\overline{\gamma}+\underline{\gamma}) \bm{g}_l$ and $\bm{g}_l := 0.5 (\overline{\gamma}-\underline{\gamma}) \bm{g}_l$.
    \end{itemize}
\end{enumerate}
\begin{remark}
    The introducing of the user-defined small positive number $\epsilon$ is to ensure  $U_{k+1}^\prime \cap \partial O_{k+1}^{(i)} = \emptyset$.
\end{remark} 
\begin{remark}
    In practice, it is a common case that $[\underline{\alpha_l}, \overline{\alpha_l}]=[-1,1]$, thus the number of generators of inner-approximation candidate $U_{k+1}^\prime$ is usually less than $O(h;U_k)$'s, which shows that this contraction method has the advantage for zonotope order reduction \cite{yang2018comparison}.
\end{remark} 
Appendix \ref{app:extra examples} provides an example (Example \ref{ex2}) to illustrate the procedure of the contraction method and why is necessary to sort the generators $\{\bm{g}_l\}_{1\leq l\leq \cols(\bm{G}_u)}$ of $U_{k+1}^\prime$ according to the angle with $\mathtt{AT}\left(\partial O_{k+1}^{(i)}\right)$.

\subsubsection{Verification of inner-approximation candidate}
According to Theorem 1 and 3 in \cite{kochdumper2020computing}, after obtaining inner-approximation candidate $U_{k+1}^\prime$, it's crucial to check whether the outer-approximation $O(h;\bm{c})$ ($\bm{c}$ is the center of $U'_{k+1}$) of the time-inverted system $\dot{\bm{x}} = -\bm{f}(\bm{x})$ is within $U_k$, which confirms the correctness of computed inner-approximation $U_{k+1}$. Since both $U_{k+1}^\prime$ and $O(h;\bm{c})$ are zonotopes, this verification reduces a zonotope containment problem. In our approach, we leverage a sufficient condition outlined in \cite{sadraddini2019linear}, which can be encoded into LP to perform the inclusion verification.

%% file: experiment.tex
\section{Experiments}
\label{experiment}
In this section we demonstrate the performance of our approach on various benchmarks. Our implementation utilizes the floating point linear programming solver GLPK and C++ library Eigen. We adopt the approach outlined in  \cite{althoff2008reachability} to compute outer-approximations appeared in our method. All  experiments herein are run on Ubuntu 20.04.3 LTS in virtual machine 
with CPU 12th Gen Intel Core i9-12900K × 8  and RAM 15.6 GB. 

To evaluate the precision of the computed inner-approximations, we use the minimum width ration $\gamma_{min}$ similar to \cite{kochdumper2020computing}, which is defined as
\begin{equation}
\label{gamma}
    \begin{aligned}
        &\gamma_{\min }=\min _{\bm{v} \in \mathcal{V}} \frac{\left|\gamma_i(\bm{v})\right|}{\left|\gamma_o(\bm{v})\right|} \\
\text{with}~~&\gamma_i(\bm{v})  =\max _{x \in U_k} \bm{v}^{\intercal} x+\max _{x \in U_k}-\bm{v}^{\intercal} x \\ 
&\gamma_o(\bm{v})  =\max _{x \in O_k} \bm{v}^{\intercal} x+\max _{x \in O_k}-\bm{v}^{\intercal} x
    \end{aligned}
\end{equation}
where $U_k$ and $O_k$ are the  inner-approximation and outer-approximation of the reachable set at $k$ step respectively. $\bm{v} \in \mathcal{V} \subset \mathbb{R}^n$, and $\mathcal{V}$ is the set consisting of $n$ axis-aligned unit-vectors. To ensure a fair comparison, the $O_k$ is chosen to be the interval enclosure of $1000$ random points at the final time instant simulated via ode45 in MATLAB. Intuitively, the larger this ratio, the better the approximation. 

Our approach is systematically compared with the state-of-the-art method presented in \cite{kochdumper2020computing}, which is publicly available in the reachability analysis toolbox CORA \cite{althoff2015introduction}. Benchmarks with system's dimension from $2$ to $12$ are utilized to show the the comprehensive advantages of our approach. Their configurations including dimensions, initial sets and references are listed in Table \ref{benchmark}. 




\begin{table}[htbp]
       \caption{Benchmarks and their dimensions, initial sets and references}
    \centering
    \setlength{\tabcolsep}{4mm}{
    \begin{tabular}{*{7}{c}}
  \toprule
   \textbf{Dim}&\textbf{Benchmark} &\textbf{Initial Set}& \textbf{Reference}\\
  \midrule
    2     & ElectroOsc & $\bm{c}_1 + [-0.1,0.1]^2$ & Example 3 in \cite{xue2016under}\\
    3     & Rossler &   $\bm{c}_2 + [-0.15,0.15]^3$ & Example 3.4.3 in \cite{chen2015reachability}    \\
    4     & Lotka-Volterra &  $\bm{c}_3 +[-0.2, 0.2]^4$ & Example 5.2.3 in \cite{chen2015reachability}  \\
    6     & Tank6 & $\bm{c}_4 + [-0.2,0.2]^6$&\cite{althoff2008reachability}\\
    7     & BiologicalSystemI &$\bm{c}_5+[-0.01, 0.01]^7$&Example 5.2.4 in \cite{chen2015reachability}   \\
    9     & BiologicalSystemII & $\bm{c}_6 + [-0.01, 0.01]^9$&Example 5.2.4 in \cite{chen2015reachability} \\
    12    & Tank12 & $\bm{c}_7 + [-0.2,0.2]^{12}$ & \cite{althoff2008reachability}  \\  

  \bottomrule
\end{tabular}}
    \label{benchmark}
\begin{footnotesize}
\begin{tablenotes}
\item Note: for the parameters of Tank6 and Tank12, all $A_i$ are set to $A_i=1$, and all $k_i$ are set to $k_i=0.015$, $\kappa=0.01, v =0, g= 9.81$, for Tank6 $n=6$ and for Tank12 $n=12$;
for the centers of initial sets, $\bm{c}_1 = (0, 3)^{\top}$, $\bm{c}_2 = (0.05,-8.35,0.05)^{\top} $, $\bm{c}_3 = (0.6,0.6,0.6,0.6)^{\top}$, $\bm{c}_4 = (2, 4, 4, 2, 10, 4)^{\top}$, $\bm{c}_5 = (0.1, 0.1,0.1, 0.1, 0.1, 0.1,0.1)^{\top}$, $\bm{c}_6 = (1, 1,1, 1, 1, 1,1,1,1)^{\top}$, $\bm{c}_7 = (2, 4, 4, 2, 10, 4,2,2,2,2,2,2)^{\top}$.
\end{tablenotes}
\end{footnotesize}
\end{table}


\subsection{Advantage in efficiency and precision}
 For each benchmark stated in Table \ref{benchmark}, we compute the inner-approximations at the time instant $T$ using our approach and the one in CORA. Table \ref{ordinary comparison} demonstrates the time cost and $\gamma_{min}$ for tow methods. The advantages of our approach are evident from low dimensional scenario ($2$-dimensional) to high dimensional scenario ($12$-dimensional), showcasing improved efficiency and precision, particularly in higher dimensions. Taking the benchmark Tank12 as an instance, our approach achieves nearly  $38\%$ improvement in precision while requiring only $12\%$ of the time compared to CORA. The visualization of the inner-approximations computed by our approach and CORA is illustrated in Fig. \ref{ordinary_figure} provided in Appendix \ref{app: table and figure}, together with the outer-approximations computed by CORA in this figure for sake of convenient comparison.
\begin{table}[htbp]
     \caption{Comparison between our approach and CORA for each benchmark}
    \centering
    \setlength{\tabcolsep}{3.5mm}
    {\begin{tabular}{*{7}{c}}
  \toprule
   \multirow{2}*{\textbf{Dim}}&\multirow{2}*{\textbf{Benchmark}} &\multirow{2}*{\textbf{T}}& \multicolumn{2}{c}{\textbf{Our Approach}} & \multicolumn{2}{c}{\textbf{ CORA}}  \\
  \cmidrule(lr){4-5}\cmidrule(lr){6-7}
  &&&time (s) & $\gamma_{min}$  & time (s) & $\gamma_{min}$\\
  \midrule
    2     & ElectroOsc & 2.5   & $\bm{23.56}$  & $\bm{0.88}$ & 36.50 & 0.57 \\
    3     & Rossler & 1.5   & $\bm{27.72}$ & 0.76 & 36.63  & $\bm{0.78}$ \\
    4     & Lotka-Volterra & 1     & $\bm{10.43}$ & $\bm{0.65}$ & 335.06 & 0.34 \\
    6     & Tank6 & 80    & $\bm{50.83}$  & $\bm{0.82}$ &201.05  & 0.63 \\
    7     & BiologicalSystemI & 0.2   & $\bm{1.74}$  & $\bm{0.96}$ & 125.73  & 0.90 \\
    9     & BiologicalSystemII & 0.2   & $\bm{72.47}$  & $\bm{0.95}$ & 188.25  & 0.88 \\
    12    & Tank12 & 60    & $\bm{235.88}$  & $\bm{0.77}$ & 1834.65 & 0.56 \\

  \bottomrule
\end{tabular}}
    \label{ordinary comparison}    
\end{table}

\subsection{Advantage in long time horizons}
Further, we extend the time horizon in Table \ref{ordinary comparison} and compare the performance of inner-approximation computation between our approach and CORA. As evidenced by the results in Table \ref{long time comparison}, our approach demonstrates the reliable capability to compute inner-approximations in relatively longer time horizons compared to CORA. It shows that our approach can consistently compute all inner-approximations while maintaining benign efficiency and precision. In contrast, the approach in CORA fails to compute inner-approximations for all benchmarks. The visualization of the inner-approximations computed by our approach and CORA is illustrated in Fig. \ref{longtime_figure} provided in Appendix \ref{app: table and figure}.
\begin{table}[htbp]
       \caption{Comparison between our approach and CORA for each benchmark in relatively longer time horizons}
    \centering
    \setlength{\tabcolsep}{3.2mm}{
    \begin{tabular}{*{7}{c}}
  \toprule
   \multirow{2}*{\textbf{Dim}}&\multirow{2}*{\textbf{Benchmark}} &\multirow{2}*{\textbf{T}}& \multicolumn{2}{c}{\textbf{Our Approach}} & \multicolumn{2}{c}{\textbf{ CORA}}  \\
  \cmidrule(lr){4-5}\cmidrule(lr){6-7}
  &&&time (s) & $\gamma_{min}$  & time (s) & $\gamma_{min}$\\
  \midrule
    2     & ElectroOsc & 3     & 73.76 & 0.90 & \multirow{7}[0]{*}{$\bm{-}$} & \multirow{7}[0]{*}{$\bm{-}$} \\
    3     & Rossler & 2.5   & 63.42 & 0.60 &       &  \\
    4     & Lotka-Volterra & 1.5   & 81.81 & 0.62 &       &  \\
    6     & Tank6 & 120   & 129.58 & 0.65 &       &  \\
    7     & BiologicalSystemI & 1.3   & 462.87 & 0.41 &       &  \\
    9     & BiologicalSystemII & 0.375 & 261.78 & 0.66 &       &  \\
    12    & Tank12 & 100    & 377.85 & 0.49 &       &  \\

  \bottomrule
\end{tabular}}
    \label{long time comparison}
    \begin{footnotesize}
\begin{tablenotes}
\item Note: the symbol ``$\bm{-}$'' means that in this experimental configuration CORA cannot compute inner-approximations.
\end{tablenotes}
\end{footnotesize}
\end{table}

\subsection{Advantage in big initial sets}
We also expand the initial sets as listed in Table \ref{benchmark} to highlight our advantage in computing inner-approximations from larger initial sets. For each benchmark, we set both a short and a long time instant  to compute inner-approximations using our approach and CORA. As shown in Table \ref{big init compairson}, our approach can accomplish all the inner-approximation computations while maintaining high levels of efficiency and precision. In contrast, for the short time instant scenario, the performance of CORA is worse than ours in both computation time and accuracy, and CORA fails to compute inner-approximations at long time instant for all benchmarks.
The visualization of the inner-approximations computed by our approach and CORA is illustrated in Fig. \ref{biginit_short_time_figure} and Fig. \ref{biginit_long_time_figure} provided in Appendix \ref{app: table and figure}.
\begin{table}[htbp]
       \caption{Comparison between our approach and CORA for each benchmark in big initial sets.}
    \centering
    \setlength{\tabcolsep}{1.7mm}{
    \begin{tabular}{*{8}{c}}
  \toprule
   \multirow{2}*{\textbf{Dim}}&\multirow{2}*{\textbf{Benchmark}} &\multirow{2}*{\textbf{Initial Set}}&\multirow{2}*{\textbf{T}}& \multicolumn{2}{c}{\textbf{Our Approach}} & \multicolumn{2}{c}{\textbf{ CORA}}  \\
  \cmidrule(lr){5-6}\cmidrule(lr){7-8}
  &&&&time (s) & $\gamma_{min}$  & time (s) & $\gamma_{min}$\\
  \midrule
  
 \multirow{2}[0]{*}{2} & \multirow{2}[0]{*}{ElectroOsc} & \multicolumn{1}{c}{\multirow{2}[0]{*}{ $\bm{c}_1 + 0.5\bm{I}^2$}} & 1.5    & $\bm{4.79}$ & $\bm{0.92}$ & \multicolumn{1}{r}{24.32} & \multicolumn{1}{r}{0.43} \\
      &       &       & 2    & 11.29 & 0.84 & $\bm{-}$    & $\bm{-}$ \\

 \multirow{2}[0]{*}{3} & \multirow{2}[0]{*}{Rossler} & \multicolumn{1}{c}{\multirow{2}[0]{*}{ $\bm{c}_2 +  0.5\bm{I}^3$}} & 1   & $\bm{15.88}$ & $\bm{0.59}$ & \multicolumn{1}{r}{36.54} & \multicolumn{1}{r}{0.53} \\
      &       &       & 1.5    & 24.01 & 0.58 & $\bm{-}$    & $\bm{-}$ \\          
 \multirow{2}[0]{*}{4} & \multirow{2}[0]{*}{Lotka-Volterra} & \multicolumn{1}{c}{\multirow{2}[0]{*}{ $\bm{c}_3 +  0.5\bm{I}^4$}} & 0.4    & $\bm{15.45}$ & $\bm{0.71}$ & \multicolumn{1}{r}{153.57} & \multicolumn{1}{r}{0.38} \\
      &       &       & 1    & 64.17 & 0.52 & $\bm{-}$    & $\bm{-}$ \\

 \multirow{2}[0]{*}{6} & \multirow{2}[0]{*}{Tank6} & \multicolumn{1}{c}{\multirow{2}[0]{*}{ $\bm{c}_4 + 0.5\bm{I}^6$}} & 80    & $\bm{66.83}$ & $\bm{0.60}$ & \multicolumn{1}{r}{463.28} & \multicolumn{1}{r}{0.13} \\
      &       &       & 100    & 80.91 & 0.53 & $\bm{-}$    & $\bm{-}$ \\

\multirow{2}[0]{*}{7} & \multirow{2}[0]{*}{BiologicalSystemI} & \multicolumn{1}{c}{\multirow{2}[0]{*}{ $\bm{c}_5 + 0.05\bm{I}^7$}} & 0.5    & $\bm{118.65}$ & $\bm{0.62}$ & \multicolumn{1}{r}{615.89} & \multicolumn{1}{r}{0.38} \\
      &       &       & 0.7    & 281.05 & 0.41 & $\bm{-}$    & $\bm{-}$ \\
      
    \multirow{2}[0]{*}{9} & \multicolumn{1}{r}{\multirow{2}[0]{*}{BiologicalSystemII}} & \multicolumn{1}{c}{\multirow{2}[0]{*}{  $\bm{c}_6 + 0.05\bm{I}^9$}} & 0.26  & $\bm{142.17}$ & $\bm{0.65}$ &1494.38 & \multicolumn{1}{r}{0.32} \\
          &       &       & 0.28  & 142.02 & 0.55 & $\bm{-}$    & $\bm{-}$ \\
          
    \multirow{2}[0]{*}{12} & \multirow{2}[0]{*}{Tank12} & \multicolumn{1}{c}{\multirow{2}[0]{*}{ $\bm{c}_7+0.5\bm{I}^{12}$}} & 40    & $\bm{162.46}$ & $\bm{0.73}$ & \multicolumn{1}{r}{1693.68} & \multicolumn{1}{r}{0.41} \\
          &       &       & 60    & 235.57 & 0.54 & $\bm{-}$    & $\bm{-}$ \\

  \bottomrule
\end{tabular}}

    \label{big init compairson}
        \begin{footnotesize}
\begin{tablenotes}
\item Note: the symbol ``$\bm{-}$'' means that in this experimental configuration CORA cannot compute inner-approximations. $\bm{I}^d$ denotes the box $[-1,1]^d$.
\end{tablenotes}
\end{footnotesize}
\end{table}

%% file: conclusion.tex
\section{Conclusion}
\label{conclusion}
In this paper we propose a novel approach to compute inner-approximations of reachable sets for nonlinear systems based on zonotopic boundary analysis. To enhance the efficiency and precision of the computed inner-approximations, we introduce three innovative and efficient methods, including the algorithm of extracting boundaries of zonotopes, the algorithm of tiling zonotopes for boundary refinement, and contraction strategy for obtaining inner-approximations from pre-computed outer-approximations. In comparison to the state-of-the-art methods for inner-approximation computation, our approach demonstrates superior performance in terms of efficiency and precision,  particularly within high dimensional cases. Moreover,  our proposed approach exhibits a remarkable capability to compute inner-approximations for scenarios with long time horizons and large initial sets, where the inner-approximations are usually failed to be computed by existing methods.


\section*{Acknowledgement}
This work is funded by the CAS Pioneer Hundred Talents Program, Basic Research Program of  Institute of Software, CAS (Grant No. ISCAS-JCMS-202302) and National Key R\&D Program of China (Grant No. 2022YFA1005101).

%% file: appendix.tex
\clearpage

\renewcommand\thesubsection{\Alph{subsection}}

\section*{Appendix}

\subsection{Proofs of Theorems in the Paper}
\label{app: proof}
\noindent\textbf{The proof of Theorem \ref{thm: Boundary of a Zonotope}:}

to prove the soundness of Alg. \ref{alg: boundaries of a zonotope}, we firstly give the boundary of a parallelotope in Lemma \ref{lemma:boundaries_paeallelotope}. 


\begin{lemma}[Boundary of a parallelotope]

\label{lemma:boundaries_paeallelotope}
An $n$-dimensional parallelotope $P = \left \langle \bm{c},\bm{G} \right \rangle$ has $2n$ facets which can be characterized by following form,

$$
\partial P^{(i)} = \left \langle \bm{c}+\bm{G}(\cdot,i),\bm{G}^{\langle i \rangle} \right \rangle,~ 
\partial P^{(n+i)} = \left \langle \bm{c}-\bm{G}(\cdot,i),\bm{G}^{\langle i \rangle} \right \rangle
$$

where $i\in\mathbb{N}_{[1,n]}$.
\end{lemma}

\begin{proof}

An $n$-dimensional parallelotope $P=\left \langle \bm{c},\bm{G} \right \rangle$ is actually the codomain of the linear map $\bm{h}: [-1,1]^n \mapsto P \subset \mathbb{R}^{n}$, where $\bm{h}(\bm{\alpha})=\bm{c}+\bm{G\alpha}, \bm{\alpha} \in [-1,1]^n$. In fact, $\bm{h}$ is a homeomorphism map for which it is a bijection and both the function $\bm{h}$ and the inverse function $\bm{h}^{-1}$ are continuous. According Corollary 6.7 in \cite{massey2019basic}, the function $\bm{h}$ maps interior points onto interior points and boundary points onto boundary points. For the domain $[-1,1]^n$, it's easy to obtain its boundary by fixing a certain dimension as $-1$ or $1$ while maintaining other dimensions (e.g., $-1 \times [-1,1]^{n-1}$ and $1 \times [-1,1]^{n-1}$). By mapping these facets of domain onto codomain, the boundary of the parallelotope $P$ can be expressed as following zonotopes.
$$
\partial P^{(i)} = \left \langle \bm{c}+\bm{G}(\cdot,i),\bm{G}^{\langle i \rangle} \right \rangle,~ 
\partial P^{(n+i)} = \left \langle \bm{c}-\bm{G}(\cdot,i),\bm{G}^{\langle i \rangle} \right \rangle
$$

where $i\in\mathbb{N}_{[1,n]}$.
\end{proof}

At this point, we give a proof of Theorem \ref{thm: Boundary of a Zonotope}.
\begin{proof}

Firstly considering an $n$-dimensional parallelotope $P=\left \langle \bm{c},\bm{G} \right \rangle$, where $\bm{G} = (\bm{g}_i)_{1 \leq i \leq n}$, we add a new generator $\bm{g}_{n+1}$ to its generator matrix, e.g. $\hat{\bm{G}} = (\bm{G}, \bm{g}_{n+1})$. This operation would form a new zonotope $\hat{Z}= \langle \bm{c},\hat{\bm{G}} \rangle$. For a certain facet $\partial P^{(i)} = \left \langle \bm{c}+\bm{g}_i,\bm{G}^{\langle i \rangle} \right \rangle$ of $P$, the addition of $\bm{g}_{n+1}$ may cause two cases (the two cases are also suitable for its symmetric facet $\partial P^{(n+i)} = \left \langle \bm{c}-\bm{g}_i,\bm{G}^{\langle i \rangle} \right \rangle$):

\textbf{Case 1}: $\mathtt{CP}(\bm{G}^{\langle i \rangle}) \cdot \bm{g}_{n+1} = 0$. This means that adding $\bm{g}_{n+1}$ won't change the position of $\partial P^{(i)}$, instead, it would change the shape  of $\partial P^{(i)}$ by replacing generator matrix $\bm{G}^{\langle i \rangle}$ with $(\bm{G}^{\langle i \rangle},\bm{g}_{n+1})$.

\textbf{Case 2}: $\mathtt{CP}(\bm{G}^{\langle i \rangle}) \cdot \bm{g}_{n+1} \neq 0$. This means that adding $\bm{g}_{n+1}$ would change the position of $\partial P^{(i)}$ and push it away from the center $\bm{c}$. if $\mathtt{CP}(\bm{G}^{\langle i \rangle}) \cdot \bm{g}_{i} > 0$, it represents that $\partial P^{(i)}$ is on the positive direction oriented by $\mathtt{CP}(\bm{G}^{\langle i \rangle})$ and negative direction otherwise. Taking the situation $\mathtt{CP}(\bm{G}^{\langle i \rangle}) \cdot \bm{g}_{i} > 0$ as an example, if $\mathtt{CP}(\bm{G}^{\langle i \rangle}) \cdot \bm{g}_{n+1} >$ ({resp.} $<$) $0$, $\partial P^{(i)}$ would translate along with $\bm{g}_{n+1}$ (resp. $-\bm{g}_{n+1}$) to ensure being a new facet away from the center $\bm{c}$, the new facet of the new zonotope $\hat{Z}$ would be $\left \langle \bm{c}+\bm{g}_i+\bm{g}_{n+1},\bm{G}^{\langle i \rangle} \right \rangle$ (resp. $\left \langle \bm{c}+\bm{g}_i-\bm{g}_{n+1},\bm{G}^{\langle i \rangle} \right \rangle$).

Viewing zonotope $\hat{Z}= \langle \bm{c},\hat{\bm{G}} \rangle$ is derived from parallelotope $\hat{P}= \langle \bm{c},\hat{\bm{G}}^{\langle i \rangle} \rangle$ adding a generator $\bm{g}_i$, we execute the same progress as case 1 and 2 do (If $\hat{\bm{G}}^{\langle i \rangle}$ isn't full rank, then skip), we loop all the $i\in\mathbb{N}_{[1,n]}$ then all the facets of $\hat{Z}$ can be obtained.


Now we consider a general zonotope $Z =\left \langle \overline{\bm{c}},\overline{\bm{G}} \right \rangle$. Inspired by aforementioned progress, to obtain all the facets of $Z$, firstly we collect all the potential hyperplanes, then for each hyperplane, these generators which don't lie in the hyperplane are chosen to make the hyperplane as far away from $\overline{\bm{c}}$ as possible, i.e. fix the direction of normal vector of the hyperplane as the positive direction, for each generator which isn't parallel with the hyperplane, if its direction is same as the normal vector, then add it to $\bm{c}$, else subtract it form $\bm{c}$. Notice that for a hyperplane, there are two facets which are symmetrical to $\bm{c}$, to obtain another facet, one just have to do same progress by changing the positive operation to negative operation.

\end{proof}

\noindent\textbf{The proof of Theorem \ref{thm: partition of a Zonotope}:}


In the sequel, we would utilize the concept of Minkowski sum, i.e., for two sets $\mathcal{S}_1, \mathcal{S}_2$, their Minkowski sum is defined as $\mathcal{S}_1 + \mathcal{S}_2 := \{s_1 + s_2 | s_1 \in \mathcal{S}_1,  s_2 \in \mathcal{S}_2\}$. For Minkowski sum of zonotopes, it has following properties:
\begin{enumerate}
    \item zonotopes are closed under Minkowski sum, let $Z_1 = \langle \bm{c}_1, \bm{G}_1 \rangle, Z_2 = \langle \bm{c}_2, \bm{G}_2 \rangle$, then $Z_1 + Z_2 = \langle \bm{c}_1 + \bm{c}_2, \bm{G}_1 + \bm{G}_2 \rangle$; 

\item for all zonotopes $Z_1, Z_2, Z_3$, $(Z_1 \cap Z_2) + Z_3 = (Z_1 + Z_3) \cap (Z_2 + Z_3), (Z_1 \cup Z_2) + Z_3 = (Z_1 + Z_3) \cup (Z_2 + Z_3)$.
\end{enumerate}

We would slightly abuse symbols if it doesn't lead to ambiguity in the proof. For a zonotope $Z$ and a generator $\bm{g}$, $Z+\bm{g}$ means that $Z + \langle\bm{g},\bm{0}\rangle$. Further, for a set $\mathcal{Z}$ containing zonotopes, $\mathcal{Z} + \bm{g}:= \{ Z + \bm{g} \mid Z \in \mathcal{Z}\}$.  Next we begin our proof of theorem \ref{thm: partition of a Zonotope}.    

\begin{proof}

    
    
    Since the symmetry of the facets of $Z$, its boundary matrix $\bm{B}$ would have the form below,

    $$
\bm{B} = \left[\begin{array}{llll}
~~0 &\cdots \\
~~0 &\cdots \\
\cdots &\cdots \\
-1 &\cdots\\
\cdots &\cdots \\
~1& \cdots \\
\cdots&\cdots
\end{array}\right],
$$
where the numbers of $-1$ in $1$ of $\bm{B}(\cdot,j)$ are equal since the operations in Alg. \ref{alg: boundaries of a zonotope}, Line $10$ and $12$. We denote the set of facets for which the $j$-th entry of the rows is $-1$ (resp. $1$) as $\mathcal{F}^{(j)}_{-1}$ (resp. $\mathcal{F}^{(j)}_{1}$). The number of facets is same for $\mathcal{F}^{(j)}_{-1}$ and $\mathcal{F}^{(j)}_{1}$ , meanwhile the symmetric facets of $\mathcal{F}^{(j)}_{-1}$ and $\mathcal{F}^{(j)}_{1}$ share the same generator matrix. We set $s$ to be the number of facets in $\mathcal{F}^{(1)}_{-1}$ and let $I_{i}$ be the indicator set containing the generator indexes of a facet $F_{i}$ (e.g., for zonotope $Z = \langle \bm{c}, (\bm{g}_1, \bm{g}_2,\bm{g}_3)\rangle$, $F_1 = \langle \bm{c}+\bm{g}_1, ( \bm{g}_2,\bm{g}_3)\rangle$ is one of its facets, then $I_1 = \{2,3\}$). 

Next we use mathematical induction.

\textbf{Base case}: $p = n$, the tiling of $Z$ is itself.

    \textbf{Inductive step}: suppose the conclusion is true for $p \geq n$, then next we would prove it's true for $p \geq n+1$.  
    

    for the $\bm{g}_1$,
    it connects two sets of symmetry facets $\mathcal{F}^{(1)}_{-1}$ and $\mathcal{F}^{(1)}_{1}$, moving the facets in the set $\mathcal{F}^{(1)}_{-1}$ along the direction of generator $\bm{g}_1$ with the length of $2\bm{g}_1$ would form a new zonotope $Z^{(1)} = \langle \bm{c} + \bm{g}_1,\bm{G}^{\langle1\rangle} \rangle$ 
    and the moving trajectory would form a few zonotopes, i.e., $\mathcal{P} = \{ P_i = F_i + \langle \bm{g}_1,\bm{g}_1 \rangle \mid F_i \in \mathcal{F}^{(1)}_{-1} \}$. (The progress of one iteration of tiling can be visualized in Fig. \ref{fig:tiling_proof_fig}.)



Firstly we would prove the zonotopes after one-step tiling don't have overlaps. 
\begin{itemize}
    \item[1.] $\forall P_i,P_j \in \mathcal{P}, i \neq j$, $P_i \cap P_j$ is a zonotope with rank less than or equal to $n-1$, i.e., $P_i^\circ \cap P_j^\circ = \emptyset$.
    
    $P_i \cap P_j = (F_i + \langle \bm{g}_1,\bm{g}_1 \rangle) \cap (F_j + \langle \bm{g}_1,\bm{g}_1 \rangle) = \langle \bm{g}_1,\bm{g}_1 \rangle + (F_i \cap F_j)$, $F_i \cap F_j$ is a zonotope whose rank is less than or equal to $n-2$, then $\langle \bm{g}_1,\bm{g}_1 \rangle + (F_i \cap F_j)$ is a zonotope whose rank is less than or equal to $n-1$, thus $P_i^\circ \cap P_j^\circ = \emptyset$.
    
    \item[2.] $\forall P_i \in \mathcal{P}$, $(Z^{(1)})^\circ \cap P_i^\circ = \emptyset$. 
    

    Firstly, according to Alg. \ref{alg: boundaries of a zonotope} and Theorem \ref{thm: Boundary of a Zonotope}, it's easy to get the boundary of $Z^{(1)}$,  which is $2\bm{g}_1 + \mathcal{F}^{(1)}_{-1} = \{ F_i + 2\bm{g}_1 \mid F_i \in \mathcal{F}^{(1)}_{-1} \} $ and $\mathcal{F}^{(1)}_{1}$. $F_i +2\bm{g}_1$ and $ F_i$ is two facets of $P_i$. 
    For a facet $F$, let $\bm{c}(F)$ and $\bm{G}(F)$ denote the center and generator matrix of $F$ respectively. Further,  let $\hat{\bm{G}}(F)$ represent the $n \times (n-1)$ submatrix of $F$'s generator matrix containing $n-1$ independent column vectors. 
    For the hyperplane $\mathcal{H}_i = \left\{\bm{x}\in \mathbb{R}^n~\big|~ \mathtt{CP}\left(\hat{\bm{G}}(F_i)\right)\cdot\left(\bm{x} - \bm{c}(F_i) -2\bm{g}_1\right)= 0 \right\} $ determined by $F_i + 2\bm{g}_1$, it divides $\mathbb{R}^n$ into two half-spaces $\mathcal{H}_i^+$ and $\mathcal{H}_i^-$. Let $\mathcal{H}_i^+$ be the half-space on the side pointed by the direction of $\bm{g}_1$ and $\mathcal{H}_i^-$ on the opposite side (see Fig. \ref{fig:tiling_proof_fig}), since $\mathcal{F}^{(1)}_{1}$ lies in $\mathcal{H}_i^+$ and $F_i$ lies in $\mathcal{H}_i^-$, then $Z^{(1)}$ and $P_i$ lies in $\mathcal{H}_i^+$ and $\mathcal{H}_i^-$ respectively (recall that a zonotope is a polytope which can be viewed as the intersection of half-spaces determined by its boundary \cite{althoff2010computing}).
  \end{itemize}  

\vspace{-5mm}
\begin{figure}[htbp]
    \centering
        \includegraphics[width=3in]{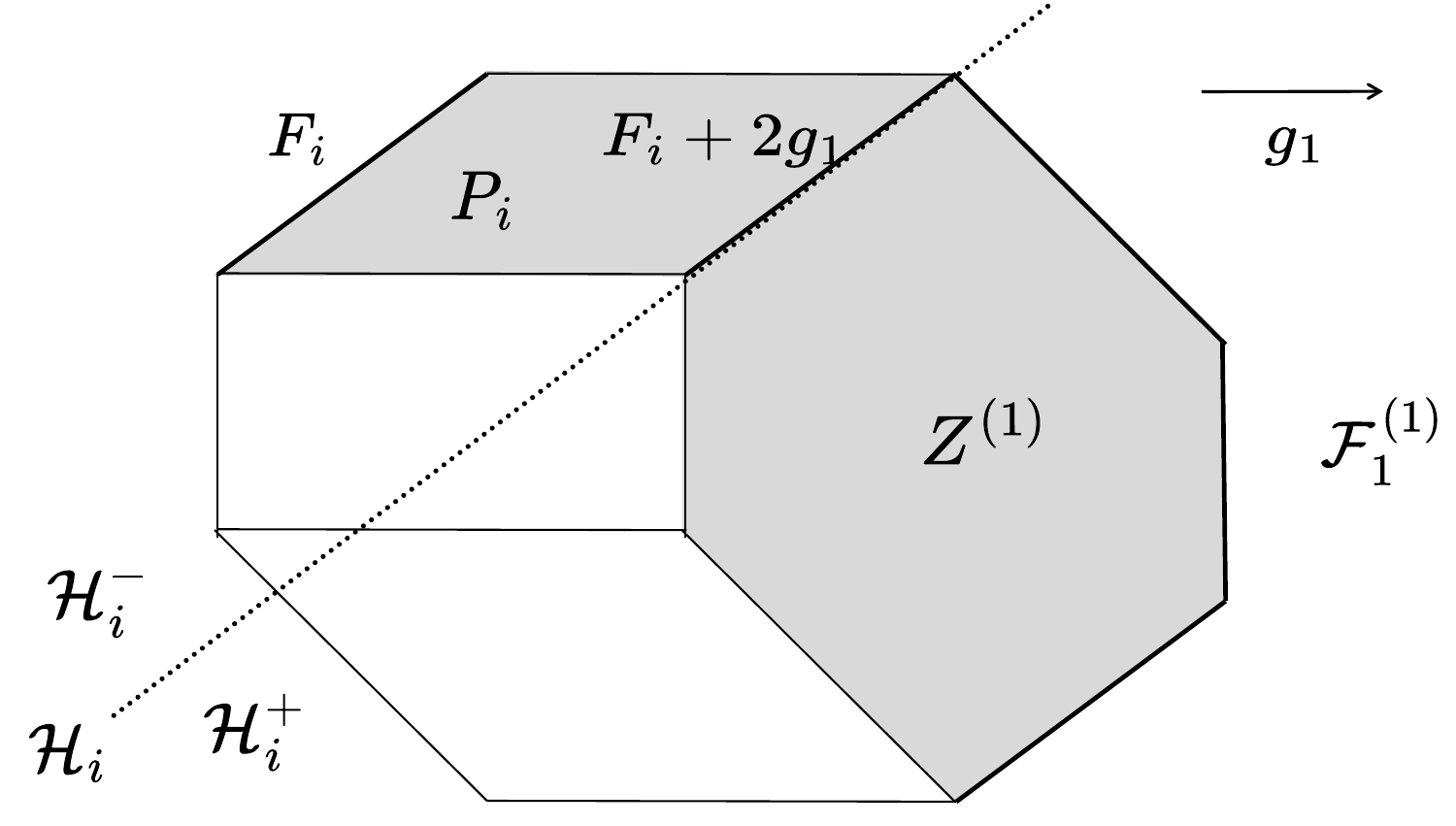}
    \caption{
    Illustration in the proof of Theorem \ref{thm: partition of a Zonotope}.}
    \label{fig:tiling_proof_fig}
\end{figure}

    Secondly we would prove that the union of zonotopes after one-step tiling is the zonotope $Z$, i.e., $Z^{(1)} \cup \bigcup\limits_{i\in \mathbb{N}{[1,s]}} P_i= Z$.

    Recalling that a volume of an $n$-dimensional zonotope with $p$ generators $Z = \langle \bm{c}, \bm{G} \rangle = \langle \bm{c}, (\bm{g}_1, \bm{g}_2, \dots, \bm{g}_p) \rangle$ is 
    \begin{equation*}
        \operatorname{vol}(Z) = 2^d \sum\limits_{1\leq j_1 < j_2 < \cdots <j_n \leq p} \left\vert \operatorname{det}\left((\bm{g}_{j_1}, \bm{g}_{j_2}, \dots, \bm{g}_{j_n})\right) \right\vert,
    \end{equation*}
    which was discovered by \cite{shephard1974combinatorial}, then 

    \begin{equation*}
    \begin{aligned}
         &\operatorname{vol}(Z^{(1)}) + \sum \limits_{i \in \mathbb{N}_{[1,s]}} \operatorname{vol}(P_i) \\
         = & 2^d \left(\sum\limits_{2\leq j_1 < j_2 < \cdots <j_n \leq p} \left| \operatorname{det}\left((\bm{g}_{j_1}, \bm{g}_{j_2}, \dots, \bm{g}_{j_n})\right) \right|\right. \\&+ \left. \sum \limits_{i \in \mathbb{N}_{[1,s]}} \sum\limits_{\substack {2\leq j_1 < j_2 < \cdots <j_{n-1} \leq p \\ j_1, j_2, \cdots,j_{n-1} \in I_i}} \left\vert \operatorname{det}\left((\bm{g}_1, \bm{g}_{j_1}, \dots, \bm{g}_{j_{n-1}})\right) \right\vert \right).
    \end{aligned}
    \end{equation*}
    Recalling Algorithm \ref{alg: boundaries of a zonotope} and Theorem \ref{thm: Boundary of a Zonotope}, the matrix set $\mathcal{G} = \{ \bm{G}(F_i) \mid F_i \in \mathcal{F}^{(1)}_{-1}\}$ is the collection all the combination of generators of $Z$ except $\bm{g}_1$ with rank $n-1$, thus $\mathcal{G}^{\prime} = \{ (\bm{g}_1, \bm{G}(F_i)) \mid F_i \in \mathcal{F}^{(1)}_{-1}\}$ collects all the combination of generators of $Z$ including $\bm{g}_1$ with rank $n$. In other words, if a combination of generators including $\bm{g}_1$ can't be found in $\mathcal{G}^{\prime}$ or isn't a submatrix of matrix in $\mathcal{G}^{\prime}$, then its rank must be less than $n$, leading to zero determinant. Therefore, we have 
      \begin{equation*}
    \begin{aligned}
    &\sum \limits_{i \in \mathbb{N}_{[1,s]}} \sum\limits_{\substack {2\leq j_1 < j_2 < \cdots <j_{n-1} \leq p \\ j_1, j_2, \cdots,j_{n-1} \in I_i}} \left\vert \operatorname{det}\left((\bm{g}_1, \bm{g}_{j_1}, \dots, \bm{g}_{j_{n-1}})\right)\right\vert  \\=& \sum\limits_{2\leq j_1 < j_2 < \cdots <j_{n-1} \leq p }\left\vert \operatorname{det}\left((\bm{g}_1, \bm{g}_{j_1}, \dots, \bm{g}_{j_{n-1}})\right) \right\vert,
    \end{aligned}
    \end{equation*}
    thus
        \begin{equation*}
    \begin{aligned}
         &\operatorname{vol}(Z^{(1)}) + \sum \limits_{i \in \mathbb{N}_{[1,s]}} \operatorname{vol}(P_i) \\
         = & 2^d \left(\sum\limits_{2\leq j_1 < j_2 < \cdots <j_n \leq p} \left| \operatorname{det}\left((\bm{g}_{j_1}, \bm{g}_{j_2}, \dots, \bm{g}_{j_n})\right) \right| \right. \\ &+  \left.\sum\limits_{2\leq j_1 < j_2 < \cdots <j_{n-1} \leq p } \left\vert \operatorname{det}\left((\bm{g}_1, \bm{g}_{j_1}, \dots, \bm{g}_{j_{n-1}})\right) \right\vert\right) \\ 
         =& 2^d \sum\limits_{1\leq j_1 < j_2 < \cdots <j_n \leq p} | \operatorname{det}\left((\bm{g}_{j_1}, \bm{g}_{j_2}, \dots, \bm{g}_{j_n})\right) | \\
         =& \operatorname{vol}(Z).
    \end{aligned}
    \end{equation*}
Besides, for all $\bm{x} \in Z^{(1)} \cup \bigcup\limits_{i\in \mathbb{N}{[1,s]}} P_i$, $\bm{x} \in Z$, then $Z^{(1)} \cup \bigcup\limits_{i\in \mathbb{N}{[1,s]}} P_i = Z$.

    After one-step tiling, the boundary matrix $\bm{B}$ represents all the facets of $Z^{(1)}$, which has $p-1$ generators, thus according to the inductive principle, the conclusion holds.

\end{proof}

\subsection{Illustrative Examples}
\label{app:extra examples}

\begin{example}
\label{ex:boundary}
Consider a $3$-dimensional zonotope with $4$ generators $Z = \left \langle \bm{c},\bm{G} \right \rangle$, where 
$$\bm{c} = \left[\begin{array}{l}
4 \\
4 \\
2
\end{array}\right], \bm{G} = (\bm{g}_1,\bm{g}_2,\bm{g}_3,\bm{g}_4) = \left[\begin{array}{llll}
1& 0& 1& 0\\
0& 1& 1& 0\\
0& 0& 0& 1
\end{array}\right].$$
We follow the computational procedures in Alg. \ref{alg: boundaries of a zonotope} to obtain all the facets of the zonotope $Z$, i.e., $\partial Z$, which is detailed below. 
\begin{enumerate}

    \item Find all the $n\times(n-1)$ submatrices of $\bm{G}$ with rank $n-1$,  
    $\mathcal{B} := \{\bm{B}_1,\bm{B}_2,\bm{B}_3, \bm{B}_4,\\ \bm{B}_5,\bm{B}_6 \} =\left\{(\bm{g}_1,\bm{g}_2),(\bm{g}_1,\bm{g}_3),(\bm{g}_1,\bm{g}_4),(\bm{g}_2,\bm{g}_3),(\bm{g}_2,\bm{g}_4),(\bm{g}_3,\bm{g}_4)\right\}$.
    \item $\bm{B} := \bm{B}_1 = (\bm{g}_1, \bm{g}_2) $; $\bm{c}_{b}^{(1)}= \bm{c}_{b}^{(2)} := \bm{c}$; $\bm{v} :=\mathtt{CP}(\bm{B}_1)=( 0,0,1)^\intercal$.
    \begin{itemize}
     \item[2a.] For $\bm{g}_3= (1,1,0)^\intercal$, due to $\bm{v} \cdot \bm{g}_3 =0$, then $\bm{B}:=(\bm{B},\bm{g}_3)=(\bm{g}_1,\bm{g}_2,\bm{g}_3)$.
        \item[2b.] For $\bm{g}_4= ( 0,0,1)^\intercal$, due to $ \bm{v} \cdot \bm{g}_4 = 1>0$, then $\bm{c}_{b}^{(1)}:=\bm{c}_{b}^{(1)}-\bm{g}_4 = ( 4,4,1)^\intercal$ and 
        $\bm{c}_{b}^{(2)}:=\bm{c}_{b}^{(2)}+\bm{g}_4 = (4,4,3)^\intercal$.      
    \end{itemize}  
    \item Now we have two facets $Z_b^{(1)} = \left \langle \bm{c}_{b}^{(1)}, \bm{B} \right \rangle$ and $\partial Z_b^{(2)} = \left \langle \bm{c}_{b}^{(2)}, \bm{B} \right \rangle$, put them into boundary set $\partial Z$. 
    \item Remove $\bm{B}_1, \bm{B}_2,\bm{B}_4$ from $\mathcal{B}$ since $\bm{B}_1, \bm{B}_2,\bm{B}_4$ are the submatrices of $\bm{B}$, now $\mathcal{B} := \{\bm{B}_3, \bm{B}_5,\bm{B}_6\}$.
    \item Repeat $(2)$-$(4)$ until $\mathcal{B} = \emptyset$, finally all eight facets of zonotope $Z$ is obtained.
\end{enumerate}
\end{example}

\begin{example}

\label{ex:boudary matrix}

Consider the zonotope $Z$ in Example \ref{ex:boundary},
which has eight facets. According to Alg. \ref{alg: boundaries of a zonotope}, we can obtain its  boundary matrix:
$$
\left[\begin{array}{llll}
~~0 &~~0& ~~0& -1 \\
~~0 &~~0 &~~0& ~~1 \\
~~0 &~~1& ~~1&~~ 0 \\
~~0 &-1& -1&~~ 0 \\
-1 &~~0& -1& ~~0 \\
~~1 &~~0 &~~1& ~~0 \\
-1& ~~1& ~~0& ~~0 \\
~~1 &-1& ~~0 &~~0 
\end{array}\right]
$$
Take the first row as an instance, the facet can be obtained as follows:
$$
\partial Z_1 =\left\langle\left[\begin{array}{l}
4 \\
4 \\
2
\end{array}\right]-\left[\begin{array}{l}
0 \\
0 \\
1
\end{array}\right],\left[\begin{array}{lll}
1& 0&1\\
0& 1&1\\
0& 0&0
\end{array}\right]\right\rangle=\left\langle\left[\begin{array}{l}
4 \\
4 \\
1
\end{array}\right],\left[\begin{array}{lll}
1& 0&1\\
0& 1&1\\
0& 0&0
\end{array}\right]\right\rangle.
$$

\end{example}

\begin{example}
\label{ex tiling}

Consider the zonotope $Z$ in Example \ref{ex:boundary} again. Its boundary matrix $\bm{B}$ is shown in Example \ref{ex:boudary matrix}.


For the rows whose the $1$-th entry equaled with $0$, delete them from boundary matrix $\bm{B}$, Then 
$$
\bm{B}=\left[\begin{array}{llll}
-1 &0& -1& 0 \\
1 &0 &1& 0 \\
-1& 1& 0& 0 \\
1 &-1& 0 &0 
\end{array}\right]
$$
 For the rows whose the $1$-th entry equaled with $-1$, add them to tiling matrix $\bm{T}$, change the value of the $1$-th entry to $0$ in tiling matrix $\bm{T}$ and and change the value of the $1$-th entry to $1$ in boundary matrix $\bm{B}$. i.e.,
$$
\bm{T} = \left[\begin{array}{llll}
0 &0& -1& 0 \\
0& 1& 0& 0 \\
\end{array}\right],
\bm{B}=\left[\begin{array}{llll}
1 &0& -1& 0 \\
1 &0 &1& 0 \\
1& 1& 0& 0 \\
1 &-1& 0 &0 
\end{array}\right]
$$
Add the last row of boundary matrix $\bm{B}$ to tiling matrix $\bm{T}$ and change the value of its $i$-th  entry ($i=2,3,4$) to $0$ in tiling matrix $\bm{T}$. At this point, we have
$$
\bm{T}=\left[\begin{array}{llll}
0 &0& -1& 0 \\
0& 1& 0& 0 \\
1 &0 &0 &0
\end{array}\right]
$$
Obtain sub-zonotopes form each row of tiling matrix $\bm{T}$ with the operations in Def. \ref{def boundary matrix}, i.e.,
$$
 Z_{1}=\left\langle\left[\begin{array}{l}
3 \\
3 \\
2
\end{array}\right],\left[\begin{array}{lll}
1& 0& 0\\
0& 1& 0\\
0& 0& 1
\end{array}\right]\right\rangle, 
 Z_{2}=\left\langle\left[\begin{array}{l}
4 \\
5 \\
2
\end{array}\right],\left[\begin{array}{lll}
1& 1& 0\\
0& 1& 0\\
0& 0& 1
\end{array}\right]\right\rangle, 
 Z_{3}=\left\langle\left[\begin{array}{l}
5 \\
4 \\
2
\end{array}\right],\left[\begin{array}{lll}
0& 1& 0\\
1& 1& 0\\
0& 0& 1
\end{array}\right]\right\rangle.
$$
The tiling of $Z$ can be visualized in Fig. \ref{fig:tiling of a zonotope}.
\end{example}

\begin{example}

\label{ex2}
Given $\partial O_{k+1}^{(1)} = \left \langle \begin{bmatrix} 1  \\  0
\end{bmatrix},
\begin{bmatrix} 1.2 & 0   \\ 0 & 0.2 
\end{bmatrix} \right \rangle$ and $O(h;U_k)= \left \langle \begin{bmatrix} 1  \\  1
\end{bmatrix},
\begin{bmatrix} 1 & 0   \\ 0 & 1 
\end{bmatrix} \right \rangle$, we compute inner-approximation candidate $U_{k+1}^\prime$ using unsorted and sorted order respectively. In this example we set $\epsilon =0.01$ for better illustration, the smallest $\epsilon$ equal to LP numerical accurate is suggested in practice.

\textbf{Unsorted case:}
    \begin{enumerate}
        \item Solving LP \eqref{min} and \eqref{max} for $\bm{g}_1=[1,0]^\intercal$, we get its domain $[\underline{\alpha_1}, \overline{\alpha_1}]=[-1,1]$, thus $\bm{g}_1$ is deleted from generator matrix $\bm{G}_u$.
        \item Solving LP \eqref{min} and \eqref{max} for $\bm{g}_2=[0,1]^\intercal$, we get its domain $[\underline{\alpha_1}, \overline{\alpha_1}]=[-1,-0.8]$, thus $[\underline{\gamma}, \overline{\gamma}] = [-0.79,1]$. Update $\bm{c}_u = \bm{c}_u + 0.5 (\overline{\gamma}+\underline{\gamma}) \bm{g}_2 = [1,1.105]^\intercal$ and $\bm{g}_2 = 0.5 (\overline{\gamma}-\underline{\gamma}) \bm{g}_2 =[0,0.895]^\intercal$.
        \item Finally, we get inner-approximation candidate $U_{k+1}^\prime = \left \langle \begin{bmatrix} 1  \\  1.105
\end{bmatrix},
\begin{bmatrix}  0   \\  0.895 
\end{bmatrix} \right \rangle$.
    \end{enumerate}
    
    \textbf{Sorted case:}   
    \begin{enumerate}
        \item For $\partial O_{k+1}^{(1)}$, its top $n-1=1$ generator by norm is $[1.2,0]^\intercal$, we compute its cross product of approximating hyperplane $\mathtt{AT}(\partial O_{k+1}^{(1)}) = [0,-1.2]^\intercal$.
        \item Sort $\{\bm{g}_1, \bm{g}_2\}$ according to $\frac{\Vert\bm{g}_l\cdot \mathtt{AT}\left(\partial O_{k+1}^{(i)}\right)\Vert}{\Vert \bm{g}_l \Vert \Vert \mathtt{AT}\left(\partial O_{k+1}^{(i)}\right) \Vert}$ from large to small, we get the new order of generators  $\{\bm{g}_2, \bm{g}_1\} = \{[0,1]^\intercal,[1,0]^\intercal\}$.
        \item Solving LP (\ref{min}) and (\ref{max}) for $\bm{g}_2=[0,1]^\intercal$, we get its domain $[\underline{\alpha_1}, \overline{\alpha_1}]=[-1,-0.8]$, thus $[\underline{\gamma}, \overline{\gamma}] = [-0.79,1]$. Update $\bm{c}_u =\bm{c}_u + 0.5 (\overline{\gamma}+\underline{\gamma}) \bm{g}_2 = [1,1.105]^\intercal$ and $\bm{g}_2 = 0.5 (\overline{\gamma}-\underline{\gamma}) \bm{g}_2 =[0,0.895]^\intercal$.
        \item Solving LP (\ref{min}) and (\ref{max}) for $\bm{g}_1=[1,0]^\intercal$, neither of them have feasible solution, thus end this process.
        \item  Finally, we get inner-approximation candidate $U_{k+1}^\prime = \left \langle \begin{bmatrix} 1  \\  1.105
\end{bmatrix},
\begin{bmatrix}  0  & 1 \\  0.895 & 0
\end{bmatrix} \right \rangle$.
    \end{enumerate}

The results of sorted and unsorted cases are visualized in figure \ref{ex2F2} and figure \ref{ex2F3} respectively. In the unsorted case $U_{k+1}^\prime$ is reduced to a $1$-dimensional zonotope while in the sorted case $U_{k+1}^\prime$ is $2$-dimensional and more bigger.

\begin{figure}[htbp]
    \centering
    \begin{subfigure}[t]{0.3\linewidth}
        \centering
        \includegraphics[width=\linewidth]{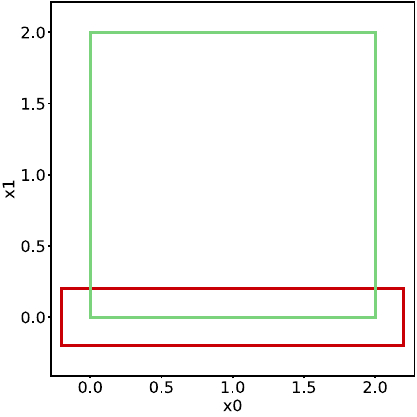}
        \caption{\textcolor{red}{$\partial O_{k+1}^{(1)}$} and initial \textcolor{mygreen}{$U_{k+1}^\prime$}}
        \label{ex2F1}
    \end{subfigure}
    \hfill
    \begin{subfigure}[t]{0.3\linewidth}
        \centering
        \includegraphics[width=\linewidth]{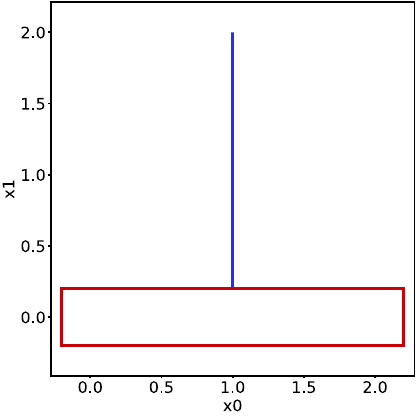}
        \caption{\textcolor{red}{$\partial O_{k+1}^{(1)}$} and final \textcolor{blue}{$U_{k+1}^\prime$} in the unsorted case}
        \label{ex2F2}
    \end{subfigure}
    \hfill
    \begin{subfigure}[t]{0.3\linewidth}
        \centering
        \includegraphics[width=\linewidth]{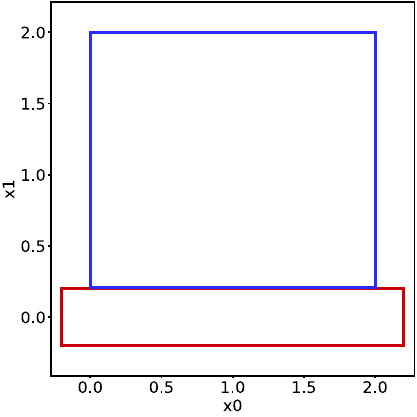}
        \caption{\textcolor{red}{$\partial O_{k+1}^{(1)}$} and final \textcolor{blue}{$U_{k+1}^\prime$} in the sorted case}
        \label{ex2F3}
    \end{subfigure}
    \caption{Illustration of Example \ref{ex2}}
    \label{fig:ex2}
\end{figure}
\end{example}

\subsection{Supplemental Figures}
\label{app: table and figure}

\begin{figure}[htbp]
    \centering
    \begin{subfigure}[t]{0.46\linewidth}
        \centering
        \includegraphics[width=\linewidth]{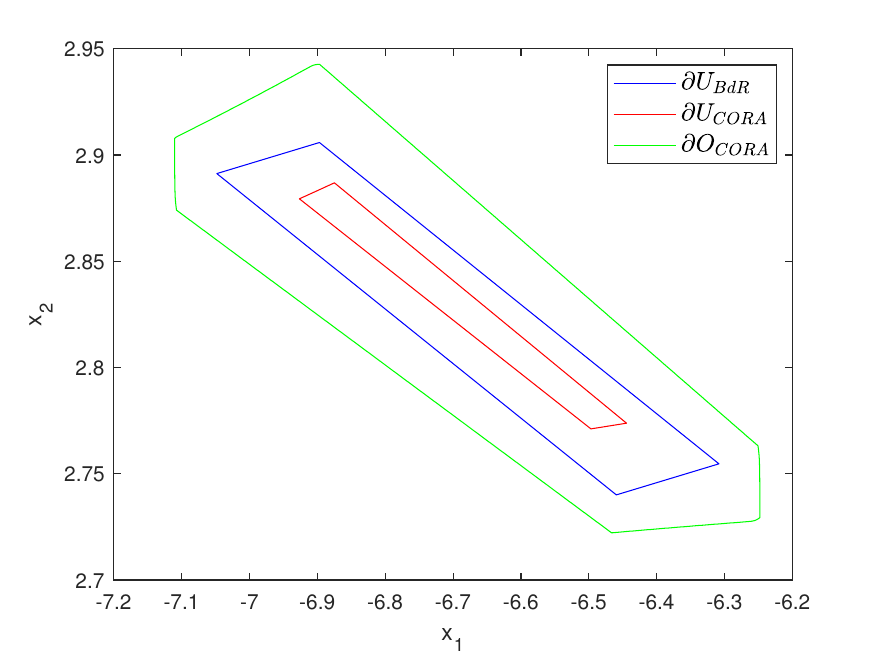}
        \caption{ElectroOsc $T = 2.5$ $\gamma_{min}:$ \textcolor{blue}{$0.88$}, \textcolor{red}{$0.57$}}
    \end{subfigure}
    \hfill
    \begin{subfigure}[t]{0.46\linewidth}
        \centering
        \includegraphics[width=\linewidth]{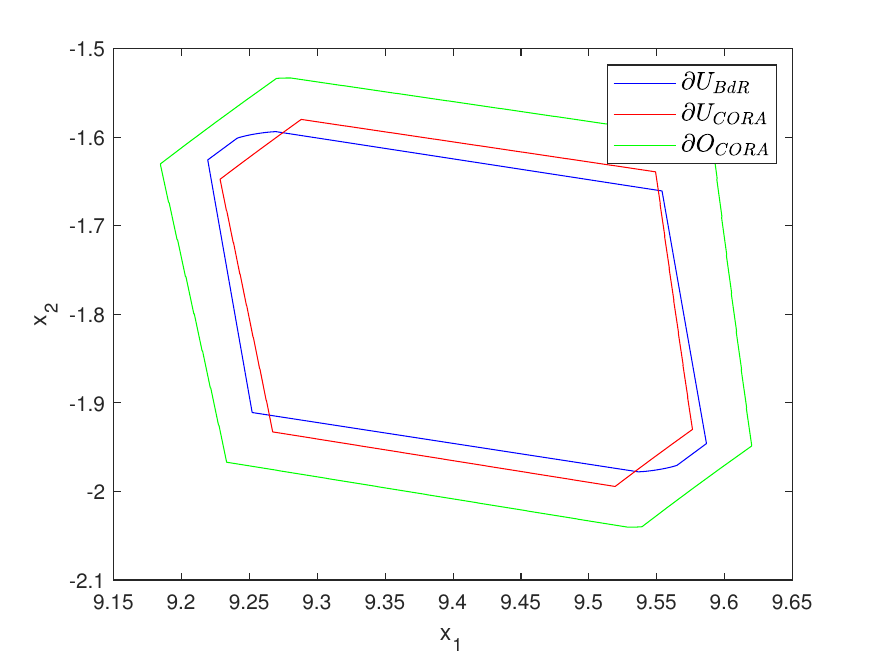}
        \caption{Rossler $T = 1.5$ $\gamma_{min}:$ \textcolor{blue}{$0.76$}, \textcolor{red}{$0.78$}}
    \end{subfigure}
    
    \bigskip

    \begin{subfigure}[t]{0.46\linewidth}
        \centering
        \includegraphics[width=\linewidth]{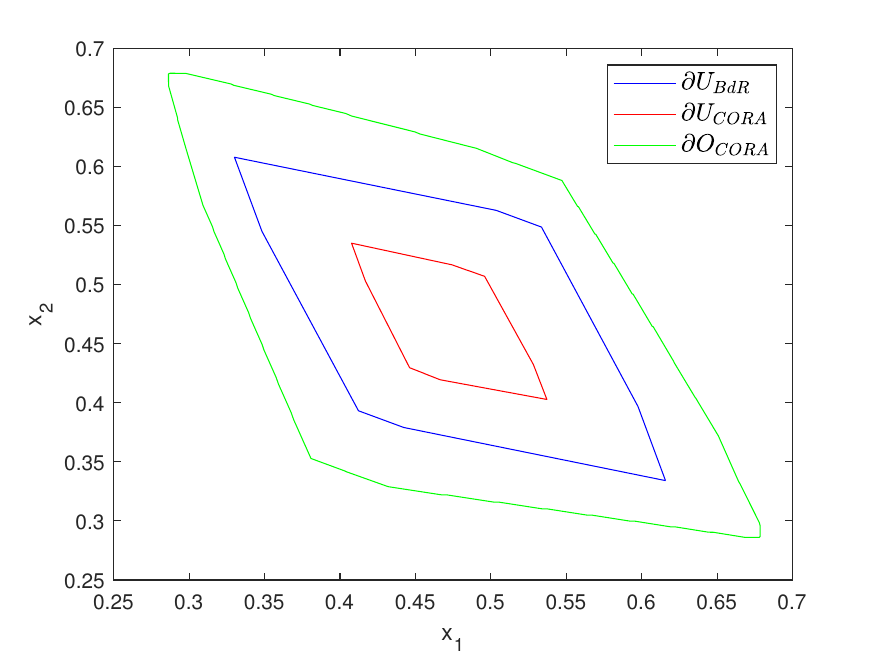}
        \caption{Lotka-Volterra $T = 1$ $\gamma_{min}:$ \textcolor{blue}{$0.65$}, \textcolor{red}{$0.34$}}
    \end{subfigure}
    \hfill
    \begin{subfigure}[t]{0.46\linewidth}
        \centering
        \includegraphics[width=\linewidth]{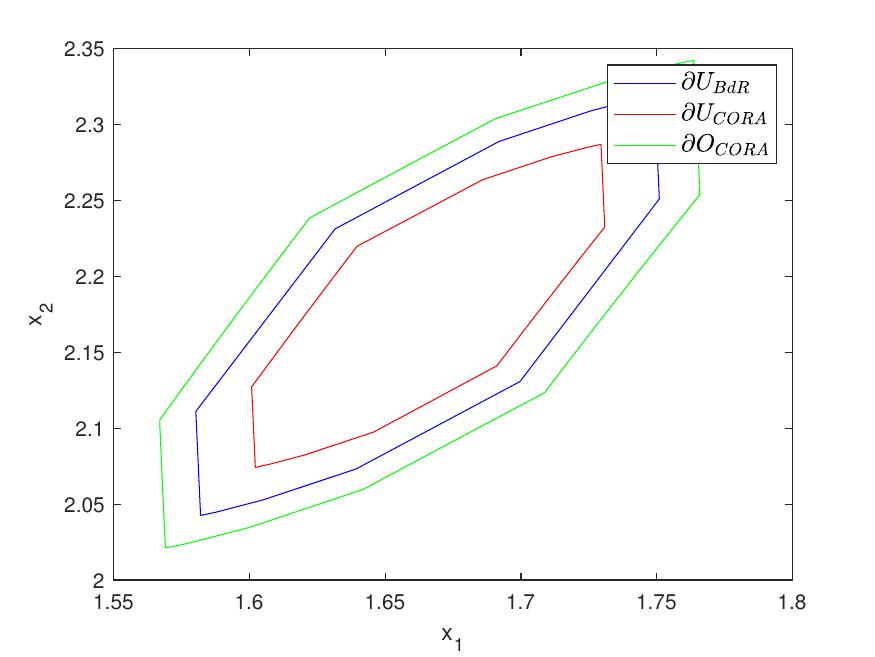}
        \caption{Tank6 $T = 80$ $\gamma_{min}:$ \textcolor{blue}{$0.82$}, \textcolor{red}{$0.63$}}
    \end{subfigure}

    \bigskip

    \begin{subfigure}[t]{0.46\linewidth}
        \centering
        \includegraphics[width=\linewidth]{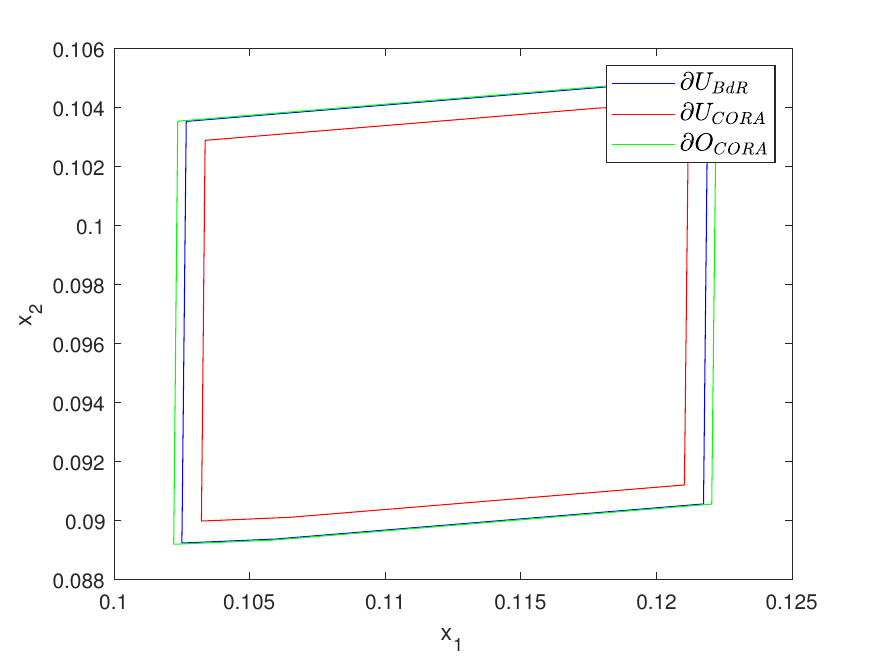}
        \caption{BiologicalSystemI $T = 0.2$ $\gamma_{min}:$ \textcolor{blue}{$0.96$}, \textcolor{red}{$0.90$}}
    \end{subfigure}
    \hfill
    \begin{subfigure}[t]{0.46\linewidth}
        \centering
        \includegraphics[width=\linewidth]{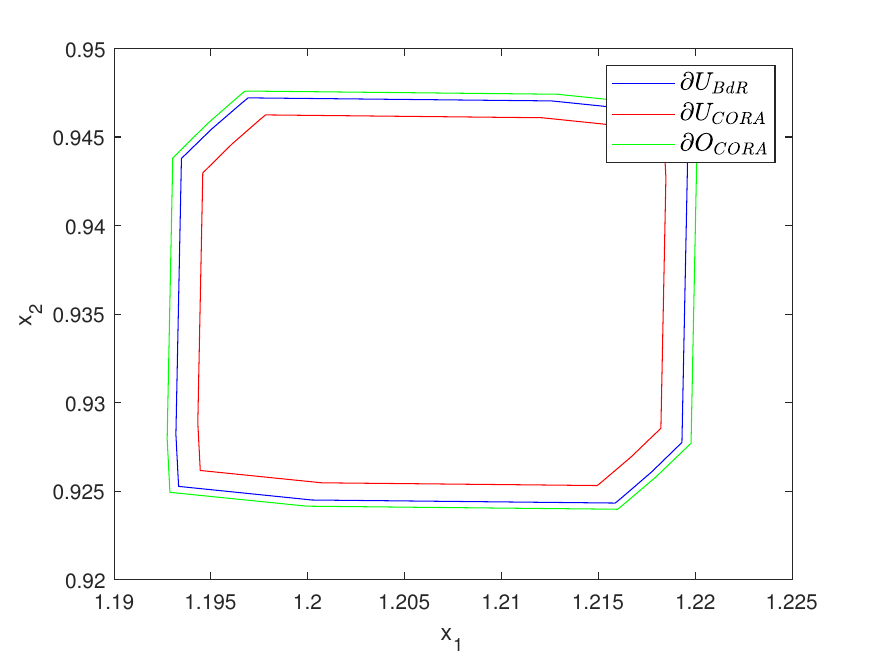}
        \caption{BiologicalSystemII $T = 0.2$ $\gamma_{min}:$ \textcolor{blue}{$0.95$}, \textcolor{red}{$0.88$}}
    \end{subfigure}

    \bigskip

    \begin{subfigure}[t]{0.46\linewidth}
        \centering
        \includegraphics[width=\linewidth]{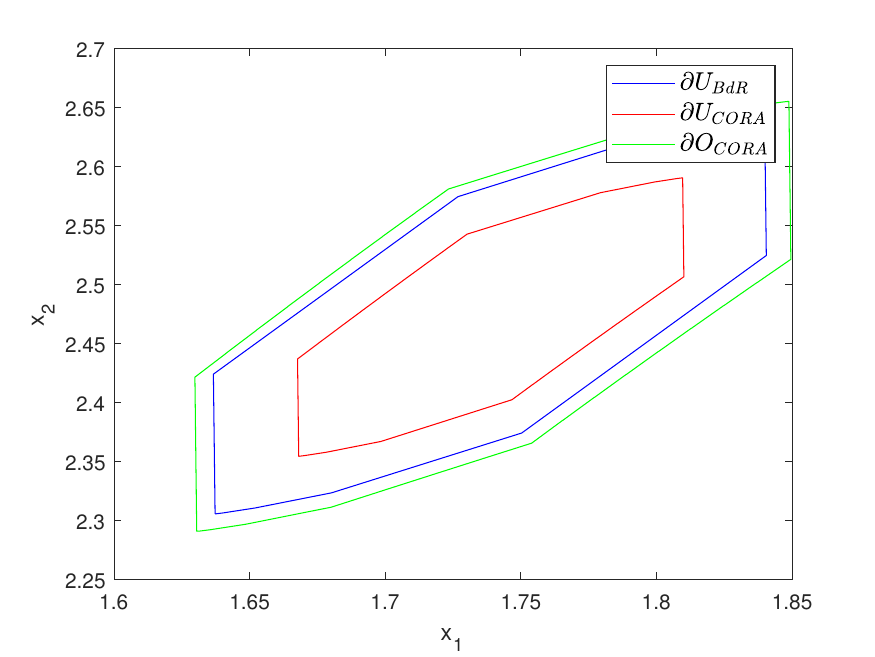}
        \caption{Tank12 $T = 60$ $\gamma_{min}:$ \textcolor{blue}{$0.77$}, \textcolor{red}{$0.56$}}
    \end{subfigure}

    \caption{Visualization of the inner-approximation computed by our approach and CORA in Table \ref{ordinary comparison}. \textcolor{blue}{Blue curve}: the boundary of inner-approximation computed by our approach. \textcolor{red}{Red curve}: the boundary of inner-approximation computed by CORA. \textcolor{green}{Green curve}: the boundary of outer-approximation computed by CORA.}
    \label{ordinary_figure}
\end{figure}

\begin{figure}[htbp]
    \centering
    \begin{subfigure}[t]{0.46\linewidth}
        \centering
        \includegraphics[width=\linewidth]{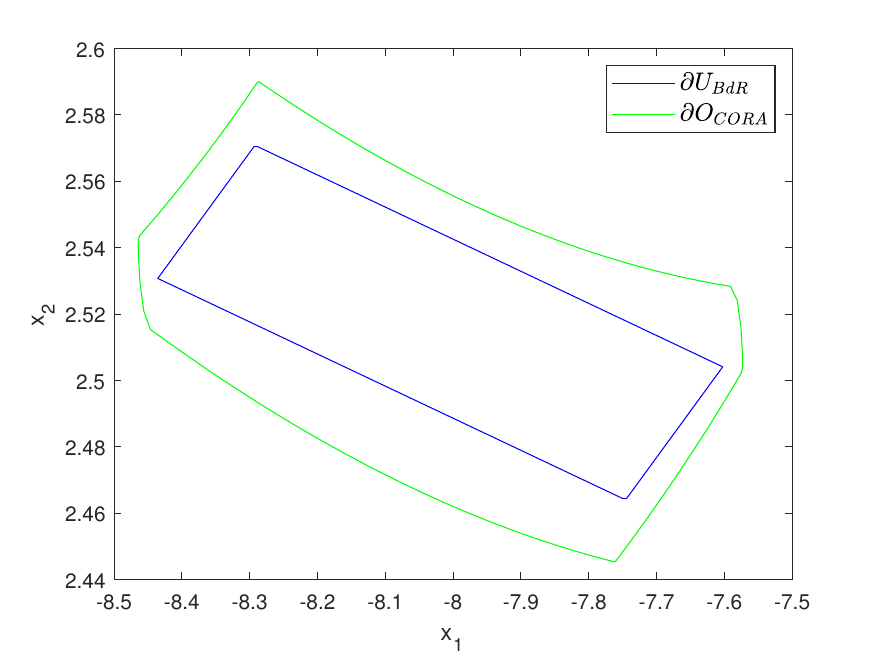}
        \caption{ElectroOsc $T = 3$ $\gamma_{min}:$ \textcolor{blue}{$0.90$}}
    \end{subfigure}
    \hfill
    \begin{subfigure}[t]{0.46\linewidth}
        \centering
        \includegraphics[width=\linewidth]{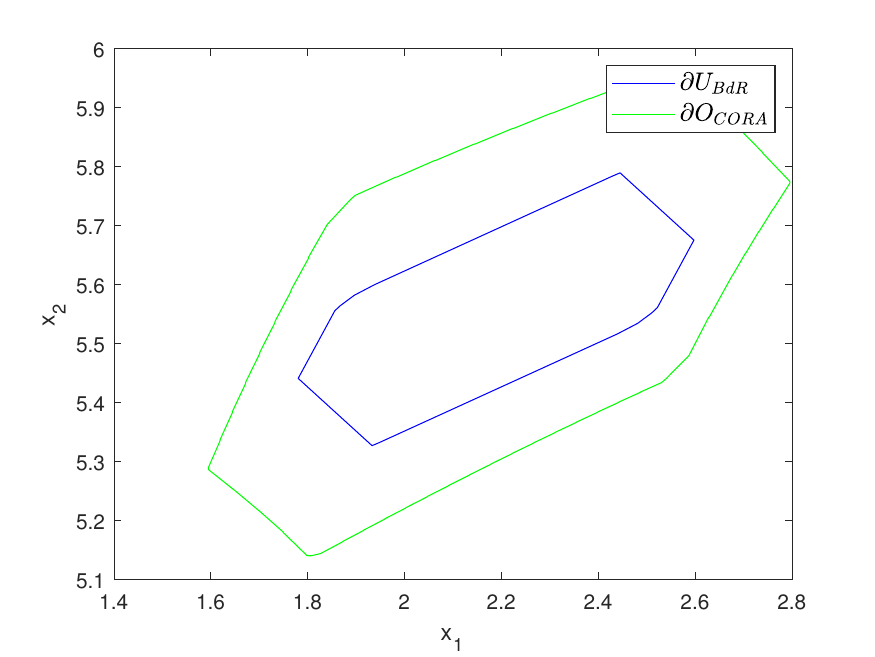}
        \caption{Rossler $T = 2.5$ $\gamma_{min}:$ \textcolor{blue}{$0.60$}}
    \end{subfigure}

    \bigskip

    \begin{subfigure}[t]{0.46\linewidth}
        \centering
        \includegraphics[width=\linewidth]{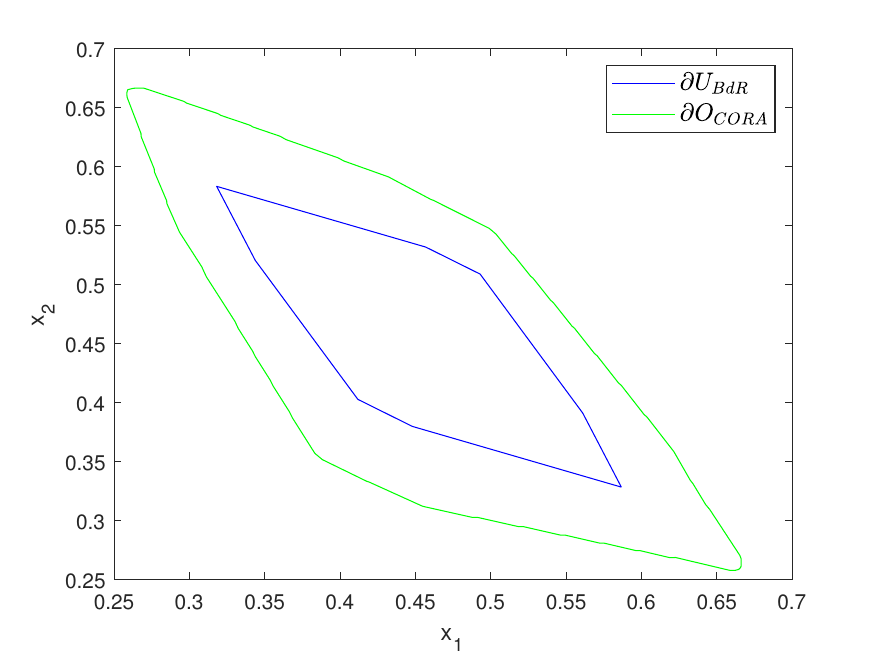}
        \caption{Lotka-Volterra $T = 1.5$ $\gamma_{min}:$ \textcolor{blue}{$0.62$}}
    \end{subfigure}
    \hfill
    \begin{subfigure}[t]{0.46\linewidth}
        \centering
        \includegraphics[width=\linewidth]{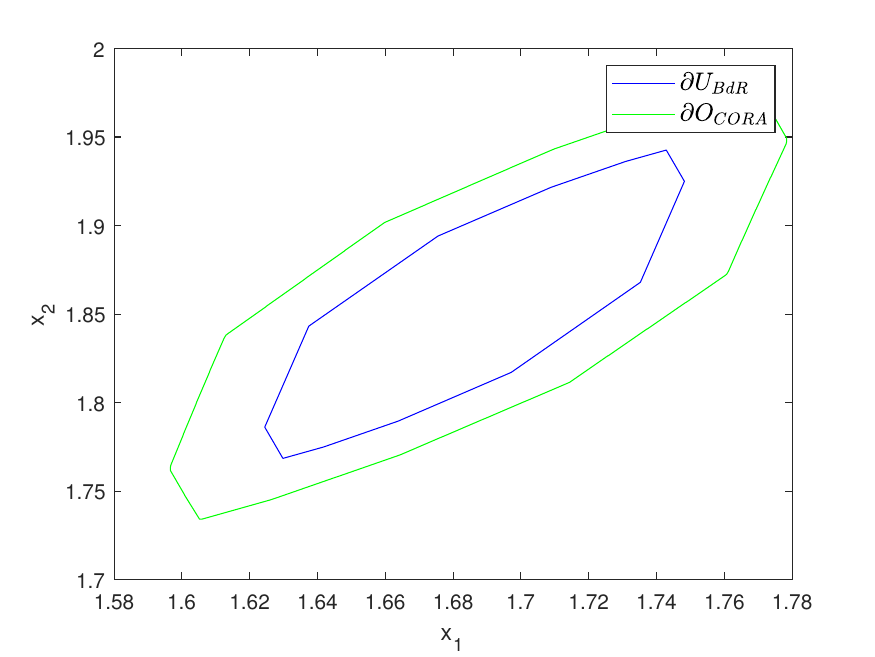}
        \caption{Tank6 $T = 120$ $\gamma_{min}:$ \textcolor{blue}{$0.65$}}
    \end{subfigure}

    \bigskip

    \begin{subfigure}[t]{0.46\linewidth}
        \centering
        \includegraphics[width=\linewidth]{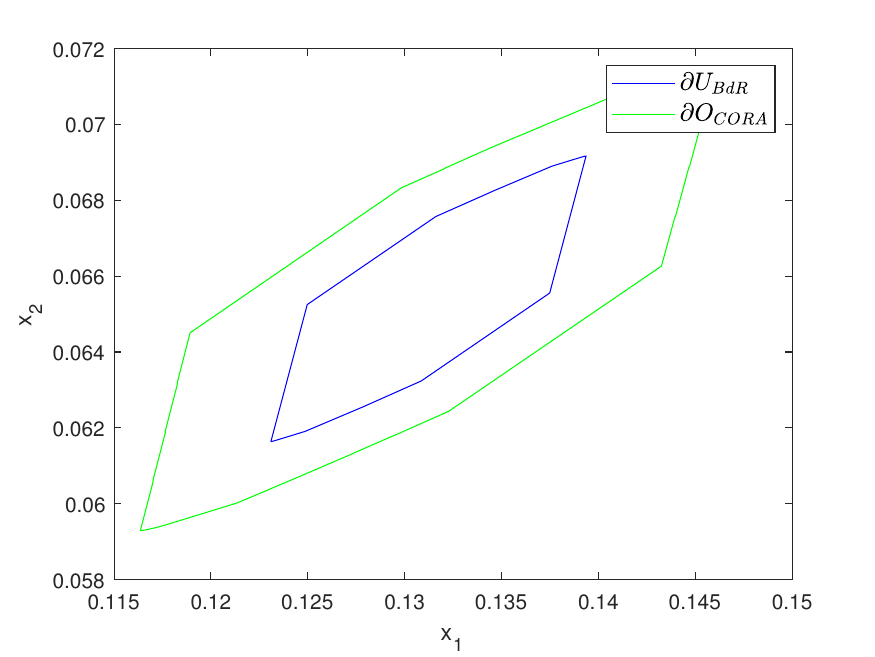}
        \caption{BiologicalSystemI $T = 1.3$ $\gamma_{min}:$ \textcolor{blue}{$0.41$}}
    \end{subfigure}
    \hfill
    \begin{subfigure}[t]{0.46\linewidth}
        \centering
        \includegraphics[width=\linewidth]{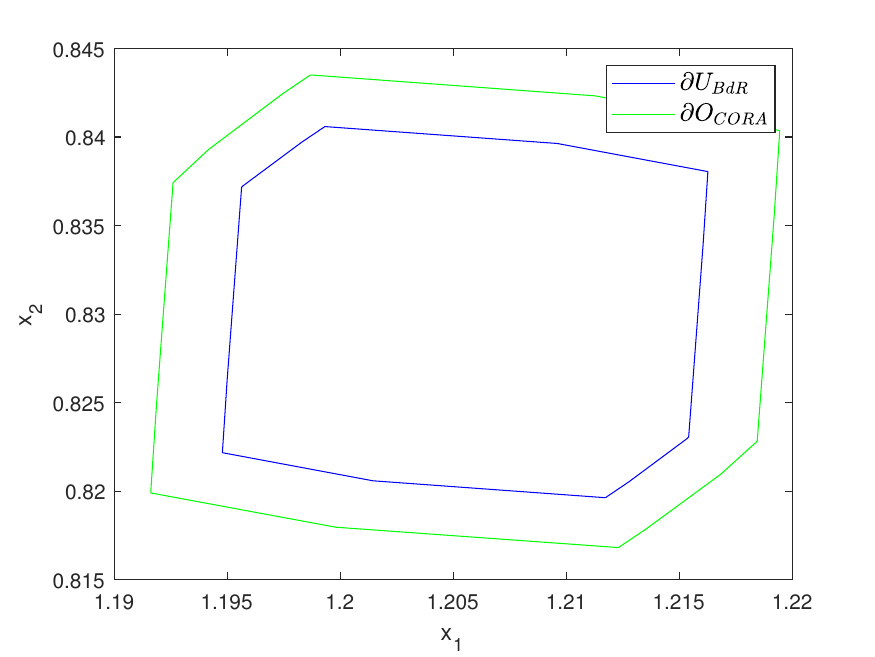}
        \caption{BiologicalSystemII $T = 0.375$ $\gamma_{min}:$ \textcolor{blue}{$0.66$}}
    \end{subfigure}

    \bigskip

    \begin{subfigure}[t]{0.46\linewidth}
        \centering
        \includegraphics[width=\linewidth]{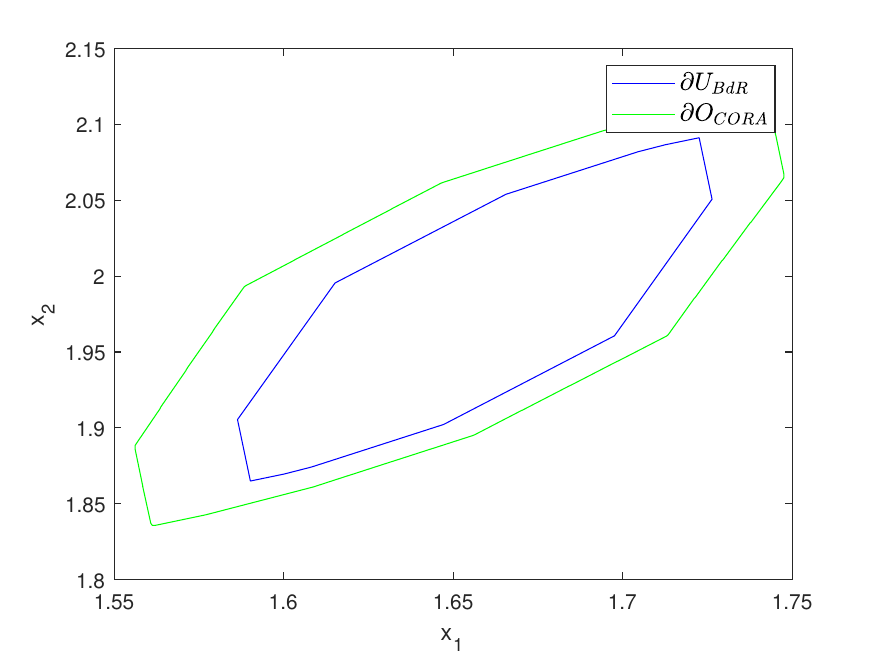}
        \caption{Tank12 $T = 100$ $\gamma_{min}:$ \textcolor{blue}{$0.49$}}
    \end{subfigure}

    \caption{Visualization of the inner-approximation computed by our approach in Table \ref{long time comparison}, for these cases CORA fails to compute inner-approximation. \textcolor{blue}{Blue curve}: the boundary of inner-approximation computed by our approach. \textcolor{green}{Green curve}: the boundary of outer-approximation computed by CORA.}
    \label{longtime_figure}
\end{figure}

\begin{figure}[htbp]
    \centering

    \begin{subfigure}[t]{0.46\linewidth}
        \centering
        \includegraphics[width=\linewidth]{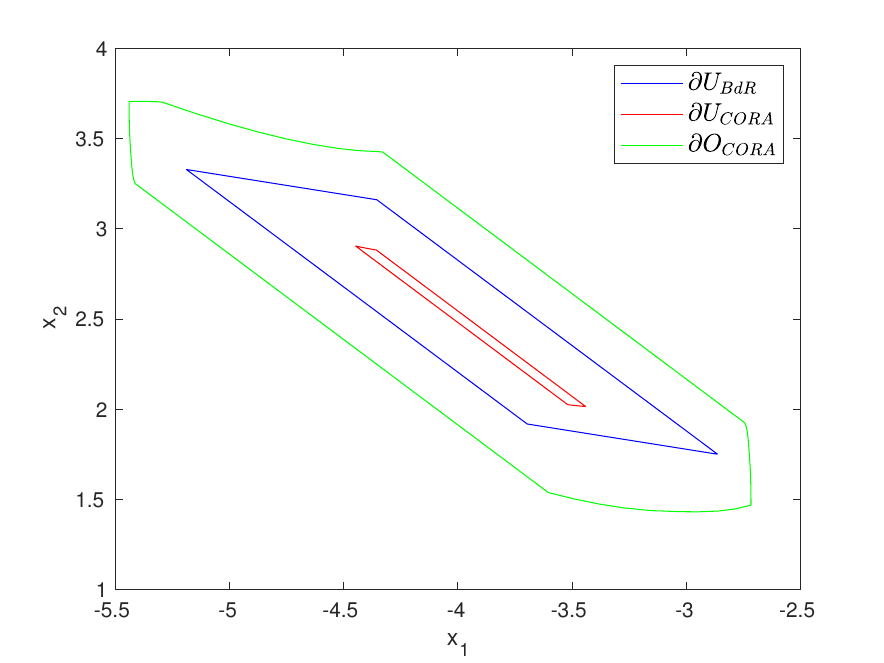}
        \caption{ElectroOsc $T = 1.5$ $\gamma_{min}:$ \textcolor{blue}{$0.92$}, \textcolor{red}{$0.43$}}
    \end{subfigure}
    \hfill
    \begin{subfigure}[t]{0.46\linewidth}
        \centering
        \includegraphics[width=\linewidth]{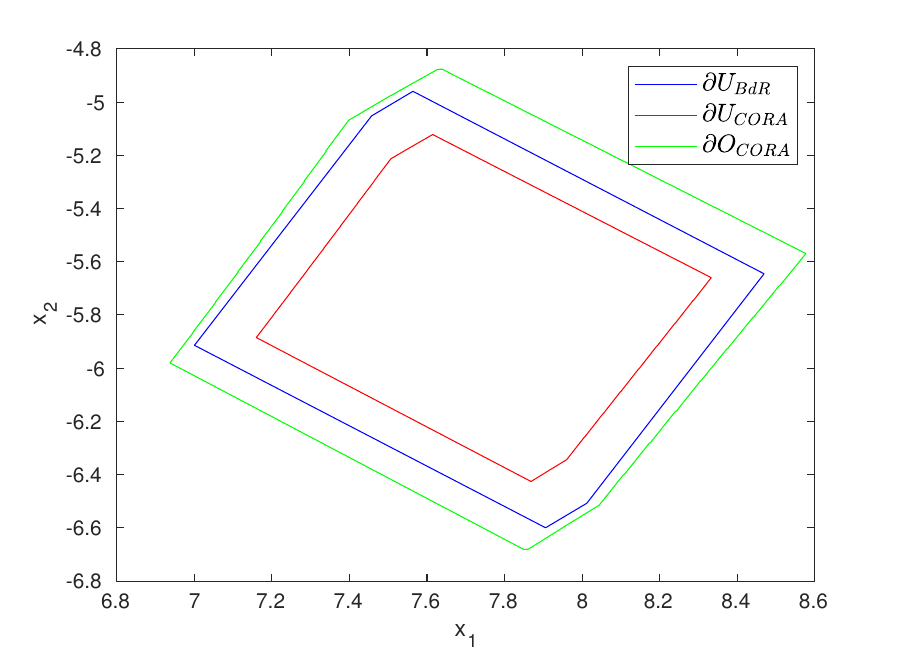}
        \caption{Rossler $T = 1$ $\gamma_{min}:$ \textcolor{blue}{$0.59$}, \textcolor{red}{$0.53$}}
    \end{subfigure}

    \bigskip

    \begin{subfigure}[t]{0.46\linewidth}
        \centering
        \includegraphics[width=\linewidth]{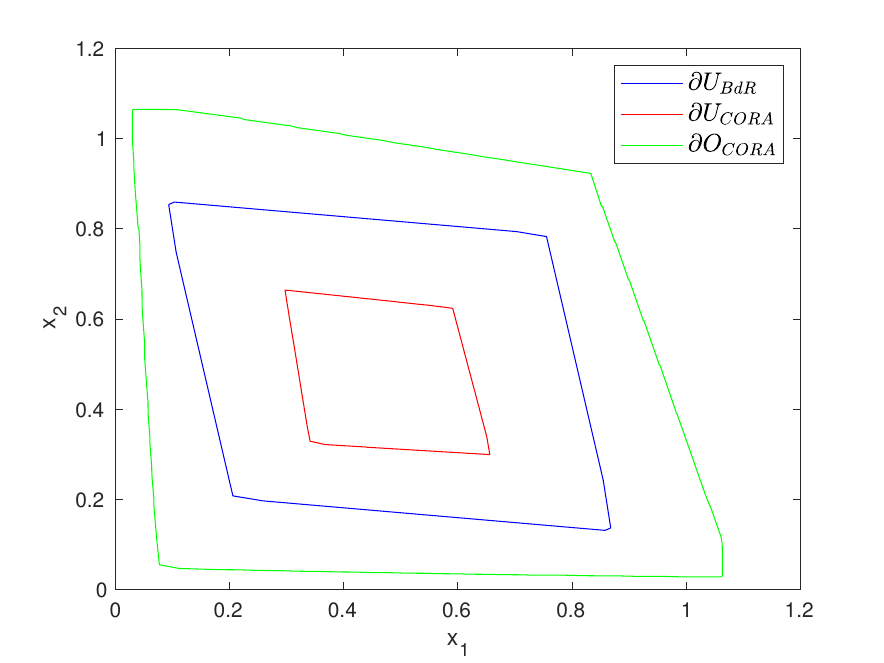}
        \caption{Lotka-Volterra $T = 0.4$ $\gamma_{min}:$ \textcolor{blue}{$0.71$}, \textcolor{red}{$0.38$}}
    \end{subfigure}
    \hfill
    \begin{subfigure}[t]{0.46\linewidth}
        \centering
        \includegraphics[width=\linewidth]{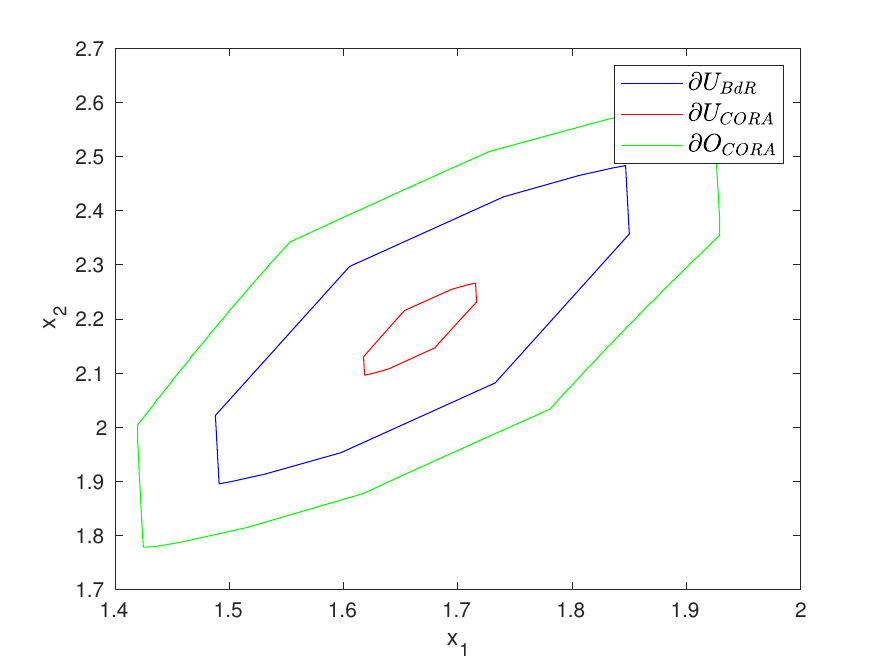}
        \caption{Tank6 $T = 80$ $\gamma_{min}:$ \textcolor{blue}{$0.60$}, \textcolor{red}{$0.13$}}
    \end{subfigure}

    \bigskip

    \begin{subfigure}[t]{0.46\linewidth}
        \centering
        \includegraphics[width=\linewidth]{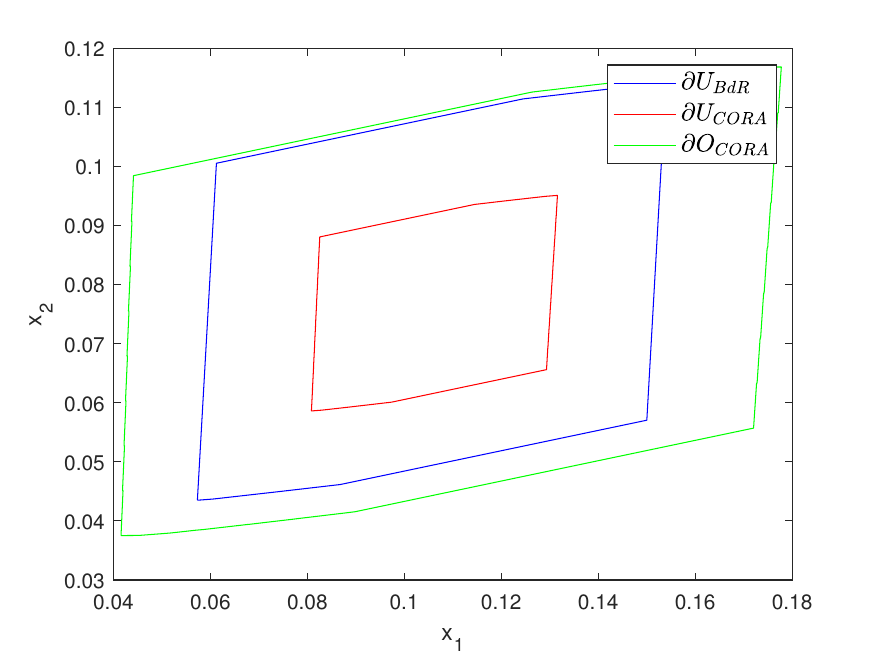}
        \caption{BiologicalSystemI $T = 0.5$ $\gamma_{min}:$ \textcolor{blue}{$0.62$}, \textcolor{red}{$0.38$}}
    \end{subfigure}
    \hfill
    \begin{subfigure}[t]{0.46\linewidth}
        \centering
        \includegraphics[width=\linewidth]{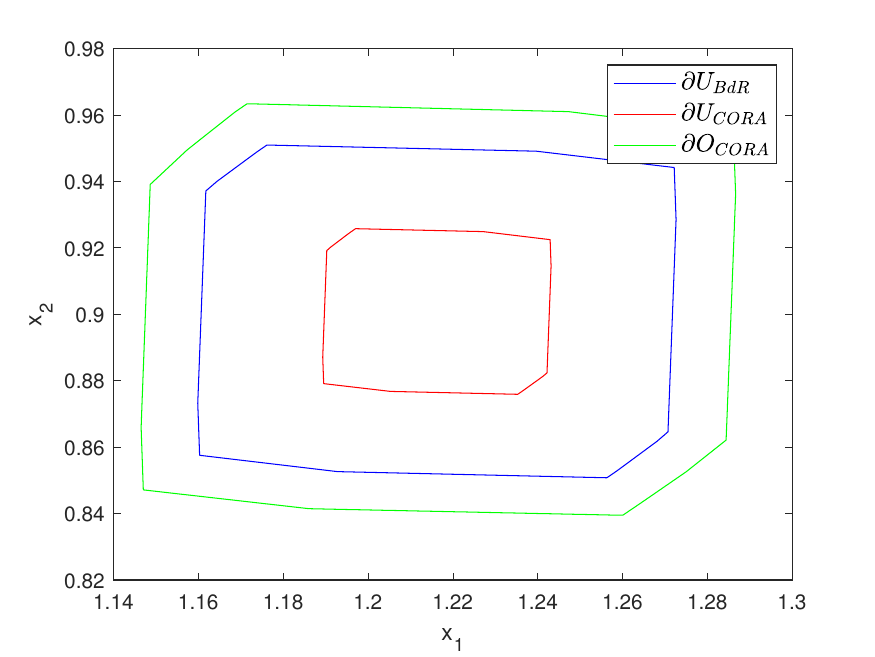}
        \caption{BiologicalSystemII $T = 0.26$ $\gamma_{min}:$ \textcolor{blue}{$0.65$}, \textcolor{red}{$0.32$}}
    \end{subfigure}

    \bigskip

    \begin{subfigure}[t]{0.46\linewidth}
        \centering
        \includegraphics[width=\linewidth]{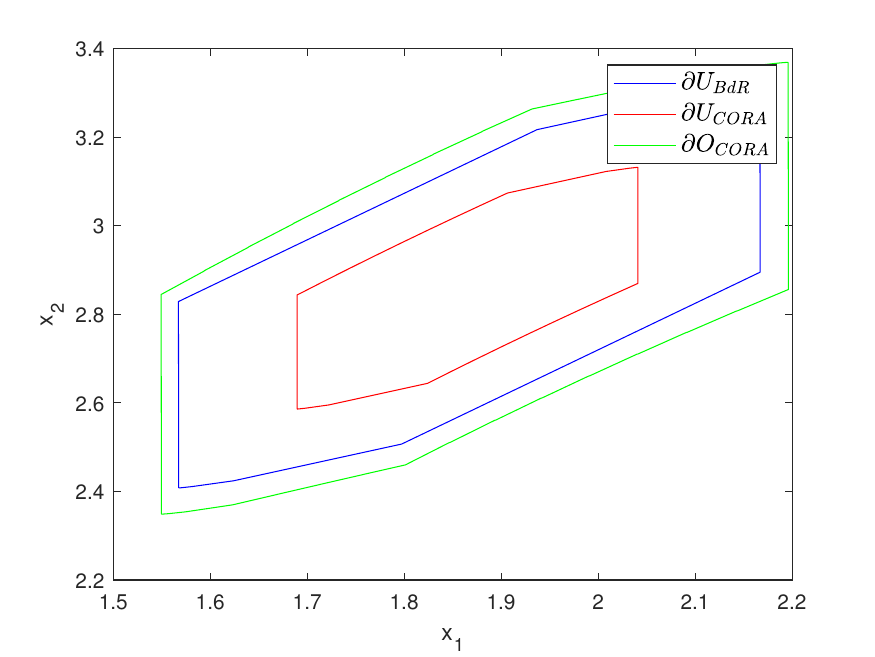}
        \caption{Tank12 $T = 40$ $\gamma_{min}:$ \textcolor{blue}{$0.73$}, \textcolor{red}{$0.41$}}
    \end{subfigure}
    
    \caption{Visualization of the inner-approximation computed by our approach and CORA in Table \ref{big init compairson}. \textcolor{blue}{Blue curve}: the boundary of inner-approximation computed by our approach. \textcolor{red}{Red curve}: the boundary of inner-approximation computed by CORA. \textcolor{green}{Green curve}: the boundary of outer-approximation computed by CORA.}
    \label{biginit_short_time_figure}
\end{figure}

\begin{figure}[htbp]
    \centering

    \begin{subfigure}[t]{0.46\linewidth}
        \centering
        \includegraphics[width=\linewidth]{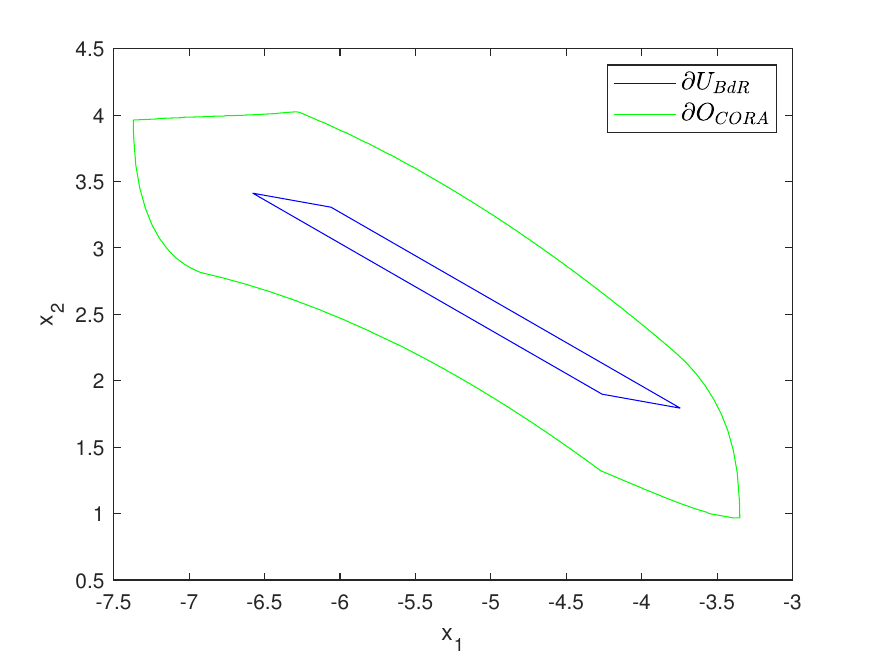}
        \caption{ElectroOsc $T = 2$ $\gamma_{min}:$ \textcolor{blue}{$0.84$}}
    \end{subfigure}
    \hfill
    \begin{subfigure}[t]{0.46\linewidth}
        \centering
        \includegraphics[width=\linewidth]{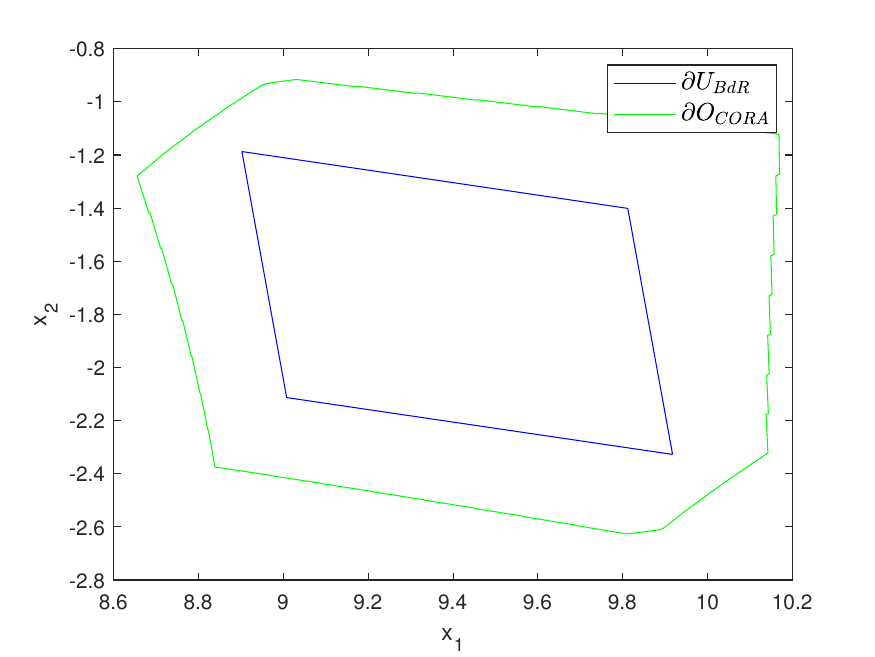}
        \caption{Rossler $T = 1.5$ $\gamma_{min}:$ \textcolor{blue}{$0.58$}}
    \end{subfigure}

    \bigskip

    \begin{subfigure}[t]{0.46\linewidth}
        \centering
        \includegraphics[width=\linewidth]{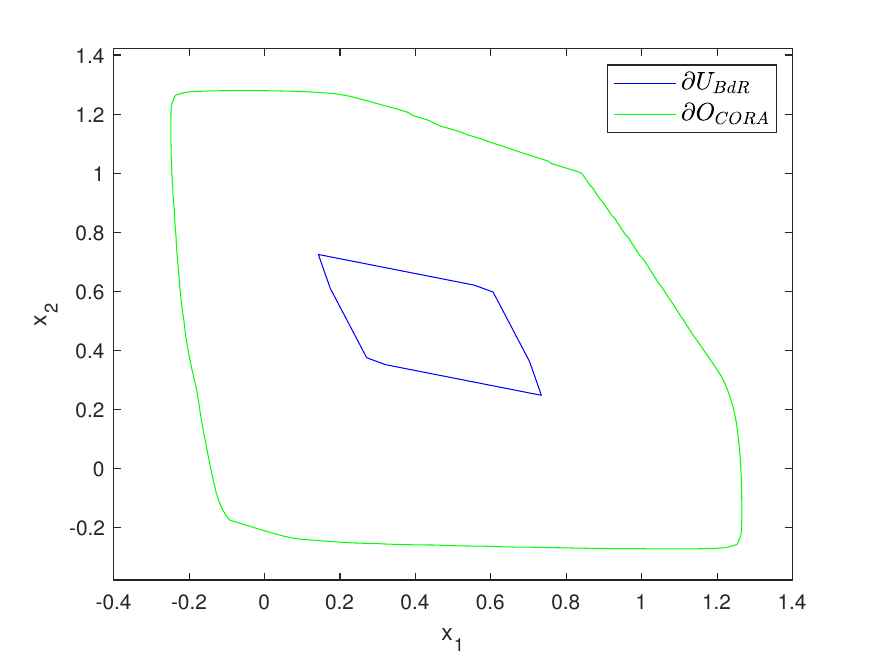}
        \caption{Lotka-Volterra $T = 1$ $\gamma_{min}:$ \textcolor{blue}{$0.52$}}
    \end{subfigure}
    \hfill
    \begin{subfigure}[t]{0.46\linewidth}
        \centering
        \includegraphics[width=\linewidth]{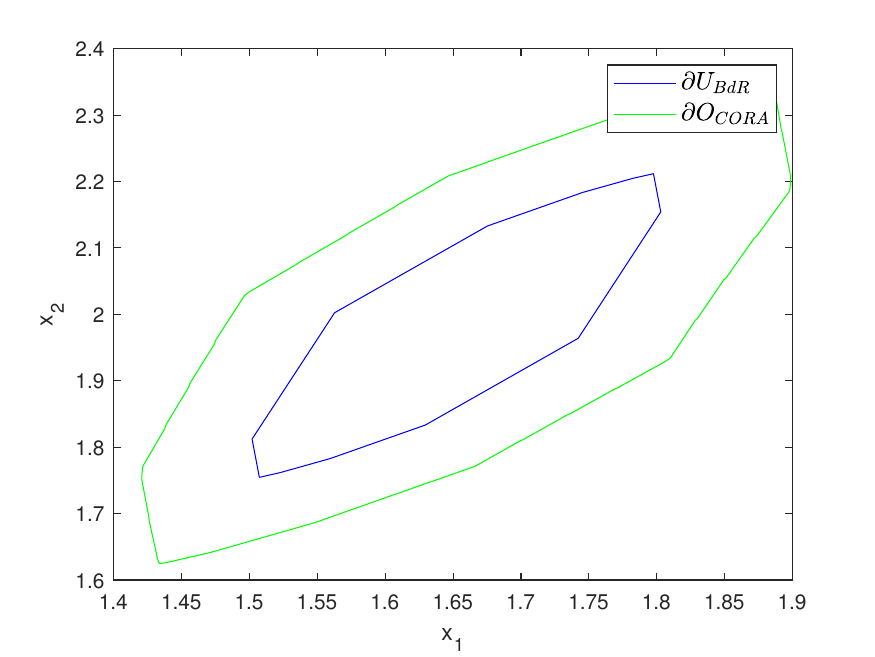}
        \caption{Tank6 $T = 100$ $\gamma_{min}:$ \textcolor{blue}{$0.53$}}
    \end{subfigure}

    \bigskip

    \begin{subfigure}[t]{0.46\linewidth}
        \centering
        \includegraphics[width=\linewidth]{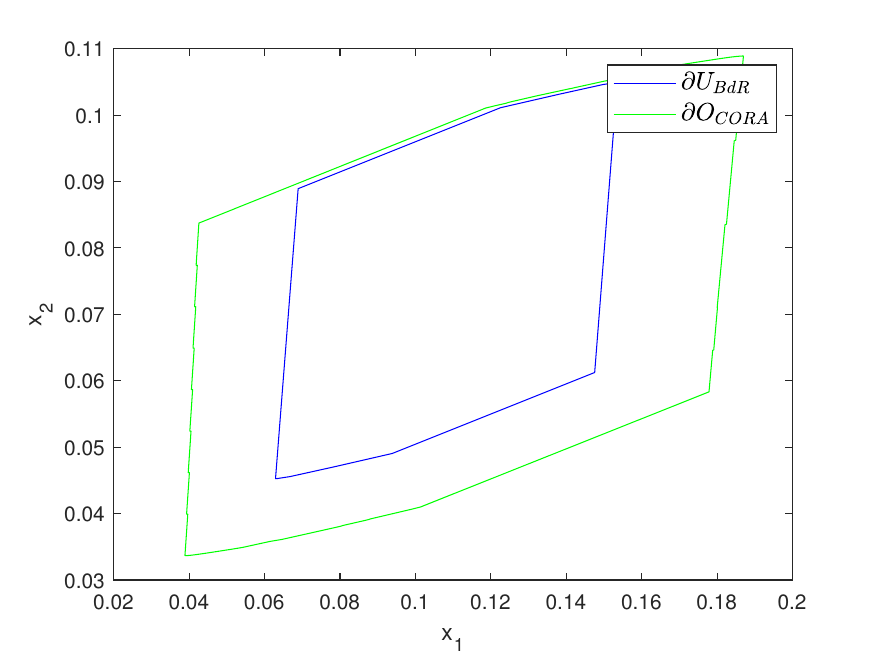}
        \caption{BiologicalSystemI $T = 0.7$ $\gamma_{min}:$ \textcolor{blue}{$0.41$}}
    \end{subfigure}
    \hfill
    \begin{subfigure}[t]{0.46\linewidth}
        \centering
        \includegraphics[width=\linewidth]{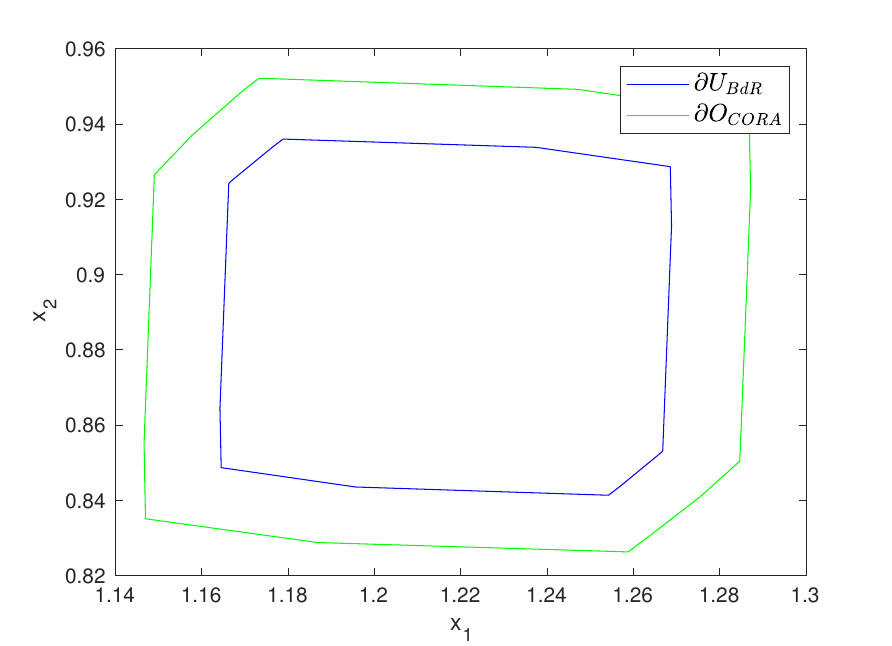}
        \caption{BiologicalSystemII $T = 0.28$ $\gamma_{min}:$ \textcolor{blue}{$0.55$}}
    \end{subfigure}

    \bigskip

    \begin{subfigure}[t]{0.46\linewidth}
        \centering
        \includegraphics[width=\linewidth]{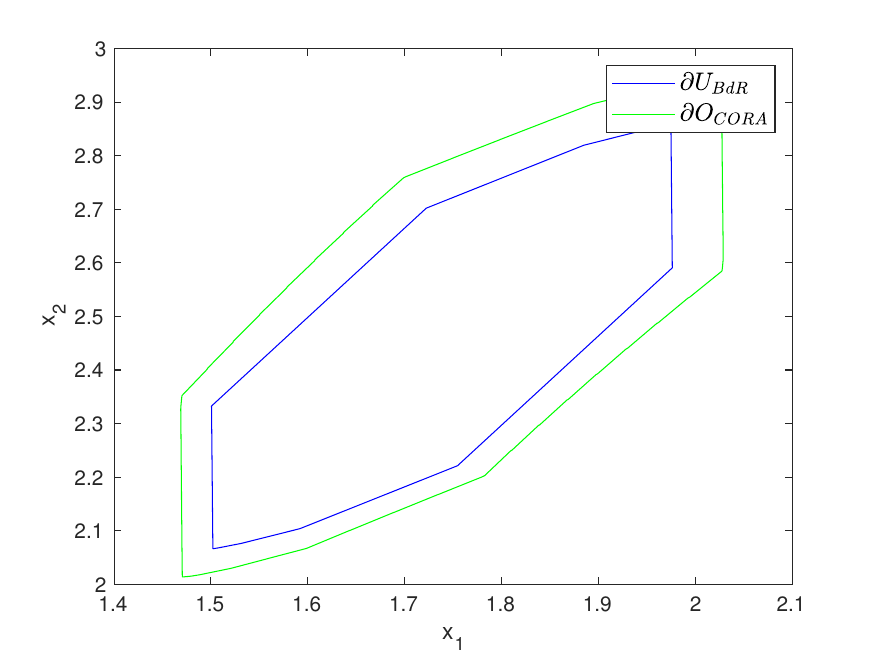}
        \caption{Tank12 $T = 60$ $\gamma_{min}:$ \textcolor{blue}{$0.54$}}
    \end{subfigure}

    \caption{Visualization of the inner-approximation computed by our approach in Table \ref{big init compairson}, for these cases CORA fails to compute inner-approximation. \textcolor{blue}{Blue curve}: the boundary of inner-approximation computed by our approach. \textcolor{green}{Green curve}: the boundary of outer-approximation computed by CORA.}
    \label{biginit_long_time_figure}
\end{figure}